\documentclass[a4paper,12pt]{article}
\usepackage{amsmath,amsthm,amsfonts,amssymb,bbm}
\usepackage{graphicx,psfrag,subfigure,color}
\usepackage{cite}
\usepackage{hyperref}

\numberwithin{equation}{section}

\newcommand{\Pb}{\mathbb{P}}

\newcommand{\dx}{\mathrm{d}}
\newcommand{\R}{\mathbb{R}}
\newcommand{\N}{\mathbb{N}}
\newcommand{\Z}{\mathbb{Z}}

\newcommand{\GOE}{\mathrm{GOE}}
\newcommand{\GUE}{\mathrm{GUE}}

\newcommand{\Or}{\mathcal{O}}

\newtheorem{assumption}{Assumption}
\newtheorem{prop}{Proposition}[section]
\newtheorem{thm}[prop]{Theorem}
\newtheorem{lem}[prop]{Lemma}

\newtheorem{cla}[prop]{Claim}

\newtheorem{rem}[prop]{Remark}

\title{Fluctuations of the competition interface in presence of shocks}
\author{Patrik L.\ Ferrari\thanks{Institute for Applied Mathematics, Bonn University, Endenicher Allee 60, 53115 Bonn, Germany. E-mail: {\tt ferrari@uni-bonn.de}} \and
Peter Nejjar\thanks{IST Austria, 3400 Klosterneuburg, Austria. E-mail: {\tt peter.nejjar@ist.ac.at}}
}

\date{December 9, 2016}

\begin{document}
\maketitle
\sloppy

\begin{abstract}
We consider last passage percolation (LPP) models with exponentially distributed random variables, which are linked to the totally asymmetric simple exclusion process (TASEP). The competition interface for LPP was introduced and studied by Ferrari and Pimentel in~\cite{PFLP05} for cases where the corresponding exclusion process had a rarefaction fan. Here we consider situations with a shock and determine the law of the fluctuations of the competition interface around its deterministic law of large number position. We also study the multipoint distribution of the LPP around the shock, extending our one-point result of~\cite{FN14}.
\end{abstract}

\newpage
\section{Introduction}
Random percolation models such as last passage percolation (LPP) and interacting particle systems like the asymmetric simple exclusion process, have been extensively studied in the past decades and they exhibit limit fluctuations laws common to random matrix theory as well (see e.g.~\cite{BDJ99,Jo00b,PS00,BR00,TW94,TW96}).

In the last passage percolation model on $\Z^2$, one assigns to each $(i,j)\in \Z^2$ an independent random variable $\omega_{i,j}\geq 0$. One can think of $\Z^2$ to have edges directed to the right and upwards. Then, given two points $A$ and $B$ which can be connected through directed paths, the basic random variable of interest in LPP is the last passage time,
\begin{equation}
L_{A\to B}=\max_{\pi:A\to B} \sum_{(i,j)\in\pi\setminus A} \omega_{i,j},
\end{equation}
where the maximum is over the set of directed paths connecting $A$ and $B$. The definition extends naturally to the case where $A$ (and/or $B$) are sets of points.

The LPP model is related with an interacting particle system, the totally asymmetric simple exclusion process (TASEP) as follows. In TASEP there are particles on $\Z$, with the exclusion constraint that one site can be occupied by at most one particle. The dynamics in continuous time is simple: particles jump to their right neighbor with a given jump rate, but jumps which would lead to a violation of the exclusion constraint are suppressed (see~\cite{Li85b,Li99} for the construction and main properties of TASEP and related models). If $\omega_{i,j}$ are taken to be waiting times of the jumps of particles and $A$ a set given in terms of the initial position of TASEP particles, then the distribution of the last passage time equals the distribution of the position of a given particle (see Section~\ref{sectLPP} for more details).

TASEP is one model in the Kardar-Parisi-Zhang (KPZ) universality class~\cite{KPZ86} in $1+1$ dimensions (see surveys and lecture notes~\cite{FS10,Cor11,QS15,BG12,Qua11,Fer10b}). The observable which is mostly studied in KPZ models is, in terms of TASEP, the joint distribution of particle positions in the limit of large time $t$. From KPZ scaling one expects that particles are correlated over a distance of order $t^{2/3}$, while the fluctuation of their position is of order $t^{1/3}$. The (conjecturally universal) limiting processes around positions where the density of particle is macroscopically smooth are also known~\cite{Jo03b,PS02,BFPS06,BFP09}.

The situation changes drastically when TASEP dynamics generates a shock, i.e., the density of particles has a discontinuity. In this situation the distribution of the rescaled particle position changes over distances of order $t^{1/3}$ instead of $t^{2/3}$ as shown in~\cite{FN14}. One result of the present paper is the extension of the findings of~\cite{FN14} to joint distributions of particle positions around the shock, using the good control over local fluctuations in LPP models established by Cator and Pimentel in~\cite{CP13}, see Section~\ref{sectMultiPts}, Theorem~\ref{mltp}.

However, the main motivation for this paper is the study of a different observable in LPP, namely the so-called \emph{competition interface} introduced by Ferrari and Pimentel in~\cite{PFLP05}. Basically they consider two lines ${\cal L}^\pm$ starting at the origin with ${\cal L}^+$ (resp.\ ${\cal L}^-$) in the second (resp.\ fourth) quadrant and color the region above ${\cal L}^+\cup {\cal L}^-$ with two colors, a point is (say) red if the LPP time from ${\cal L}^+$ to it is larger than from ${\cal L}^-$ and blue otherwise. The interface between the two colors is called the competition interface. Interestingly, the competition interface and the trajectory of the second class particle in TASEP are the same~\cite{PFLP05} (see also Section~\ref{SectSecondClass}). The importance of second class particles, is that in presence of shocks they can be used to identify it~\cite{Li99}, while for stationary initial conditions the distribution of its position is proportional to the two-point function, quantity measuring the space-time correlation~\cite{PS01,FS05a}).

If the lines ${\cal L}^\pm$ have asymptotically a fixed direction, then the direction of the competition interface converges almost surely to a value which might be deterministic or random as shown by Ferrari, Martin and Pimentel in~\cite[Theorem 2]{PFJBLP09}. In particular, when in the TASEP framework there is a shock, the competition interface satisfies a (deterministic) law of large numbers (see also~\cite{FK95}).

The main contribution of this paper is the study of the fluctuation of the competition interface in presence of shocks. In Theorem~\ref{comp2speed} we have a CLT type result for the fluctuations of the position of the competition interface with some deterministic ${\cal L}^\pm$. The difference with the standard CLT is that the fluctuations at distance $t$ from the origin are of order $t^{1/3}$ and the distribution law is not Gaussian. This result serves as illustration of Theorem~\ref{Genthm} which applies to more general LPP problems.

Finally, in Theorem~\ref{berthm} we show Gaussian fluctuations for the fluctuations of the competition interface in the case that ${\cal L}^\pm$ are random walks corresponding to TASEP with initial condition given by the product measure with density $\rho$ on $\Z_-$ and $\lambda$ on $\N$ (with $\rho<\lambda$). In that case the second class particle has also Gaussian fluctuations in the $t^{1/2}$ scale around the macroscopic shock position~\cite{Fer90,FF94b,PG90}. Although the position of the competition interface is a random time change of the position of the second class particle~\cite{PFLP05} (see also Section~\ref{SectSecondClass}), it is not straightforward to connect the two quantities when the speed of the shock is non-zero. The analogue result for the time-changed process (essentially for the second-class particle) can be found in Theorem~4 of~\cite{PFJBLP09} (see also~\cite{CP07,CFP08} for related studies for the Hammersley process).

The results obtained in this paper are based on the asymptotic independence of two LPP problems, which in turn are based on the slow-decorrelation phenomenon occurring along the characteristic lines~\cite{Fer08,CFP10b}. For TASEP with stationary initial conditions~\cite{Lig76}, i.e., Bernoulli product measure with $\rho=\lambda$, this cannot be applied anymore. This case has been however studied fairly well. Firstly, one can identify the competition interface with the maximizer of the (backwards) LPP problem~\cite{BCS06}. Secondly, the second class particle starting from the origin has some explicitly known scaling function~\cite{PS01,FS05a,BFP12}. Recent progresses on the knowledge of the time-time correlation in the KPZ height function~\cite{Dot13,Dot16,Joh15,FS16,NDT16} show that it would be relevant to study the competition interface and the second class particle as a process in time as well.

\bigskip\noindent
{\bf Acknowledgments.} The work of P.L. Ferrari is supported by the German Research Foundation via the SFB 1060--B04 project. P. Nejjar's work was mostly undertaken while he was a postdoc at Ecole Normale Sup\'{e}rieure,  D\'{e}partement de math\'{e}matiques et applications, Paris.

\section{Last Passage Percolation and Competition Interface}\label{sectLPP}
\subsection{Models and main Result}\label{SectModel}
We consider last passage percolation times on $\Z^{2}$. Fix a starting set $S_{A}\subset \Z^{2}$ and an endset $S_{E}\subset \Z^{2}.$ An up-right path from $S_{A}$ to $S_{E}$ is a sequence of points $\pi=(\pi(0), \pi(1),\ldots,\pi(n))\in \Z^{2n}$ such that $\pi(0)\in S_{A}, \pi(n)\in S_{E}$ and $\pi(i)-\pi(i-1)\in \{(0,1),(1,0)\}$. We denote by $\ell(\pi)=n$ the length of $\pi$.
Take a family $\{\omega_{i,j}\}_{i,j\in \Z} $ of independent, nonnegative random variables, and define the last passage percolation (LPP) time from $S_A$ to $S_E$ to be
\begin{equation}
\label{LPP}
L_{S_{A}\to S_{E}}:=\max_{\pi: S_A\to S_E} \sum_{(i,j)\in\pi\setminus S_A}\omega_{i,j}
\end{equation}
and define, say, $L_{S_{A}\to S_{E}}$ to be $-\infty$ when the maximum in $\eqref{LPP}$ does not exist. We denote by $\pi^{\mathrm{max}}$ any path for which the maximum
in \eqref{LPP} is attained and call $\pi^{\mathrm{max}}$ a maximizing path. In the LPP models we consider in this paper, $\pi^{\mathrm{max}}$ is always a.s.\ unique.
Last passage percolation is linked to the totally asymmetric simple exclusion process (TASEP) in the following way.
In TASEP, each $i \in \Z$ is occupied by at most one particle, and each particle waits an exponential time (whose parameter may depend on the particle) before it jumps one step to the right iff its right neighbor is not occupied. Labeling particles from right to left
\[\ldots < x_2(0) < x_1(0) < 0 \leq x_0(0)< x_{-1}(0)< \cdots \]
we denote by $x_n (t)$ the position of particle $n$ at time $t$ and
 have that at all time $t\geq 0$, $x_{n+1}(t)<x_n(t)$, $n\in\Z$. TASEP can be translated into a LPP model as follows. Define
\begin{equation} \label{startLPP}
\mathcal{L}= \{(k+x_{k}(0), k):k \in \Z)\}
\end{equation} and weights
\begin{equation}\label{startweights}
\omega_{i,j}:=\exp(v_j)
\end{equation}
where $v_j$ is the parameter of the exponential clock attached to particle $j$. Then we have
\begin{equation}\label{LPPTASEP}
\Pb\bigg(\bigcap_{k=1}^{\ell} \{x_{n_k} (t) \geq m_k-n_k\}\bigg)=\Pb\bigg(\bigcap_{k=1}^{\ell} \{L_{\mathcal{L}\to (m_k,n_k)}\leq t\}\bigg).
\end{equation}
It will be important for us to consider separately the sets
\begin{equation}\label{Lpm}
\mathcal{L}^{+}=\{(k+x_{k}(0), k):k >0\} \quad \textrm{and}\quad \mathcal{L}^{-}=\{(k+x_{k}(0), k):k \leq 0\}.
\end{equation}
We denote by $\pi^{\max}_{+},\pi^{\max}_{-}$ the (a.s.\ unique) maximizing paths of $L_{\mathcal{L}^{+}\to (m,n)},L_{\mathcal{L}^{-}\to (m,n)}$. The aim of this section is to obtain limit results for the competition interface in the LPP model \eqref{startLPP}, \eqref{startweights}. The competition interface is essentially the boundary of the region where the LPP to ${\cal L}^+$ is larger than the LPP to ${\cal L}^-$.

\subsubsection{The competition interface}\label{sectCompInt}
Here we study the LPP model defined by \eqref{startLPP}, \eqref{startweights} with a competition interface, which was introduced in~\cite{PFLP05}. To associate a competition interface to the LPP time, we consider TASEP with initial data satisfying $x_0 (0)= 1$ and $x_{1} (0)<-1$. Note that for such initial data we have (with $\mathcal{L}^{+}, \mathcal{L}^{-}$ defined in \eqref{Lpm}) almost surely $L_{\mathcal{L}^{+}\to (i,j)}\neq L_{\mathcal{L}^{-}\to (i,j)}$ for $(i,j) \in \N^{2}$.
We then define two clusters via
\begin{equation}
\begin{aligned}
&\Gamma_{+}^{\infty}:=\{(i,j)\in \Z^{2}:L_{\mathcal{L}^{+}\to (i,j)}>L_{\mathcal{L}^{-}\to (i,j)}\},\\
& \Gamma_{-}^{\infty}:=\{(i,j)\in \Z^{2}:L_{\mathcal{L}^{-}\to (i,j)}>L_{\mathcal{L}^{+}\to (i,j)}\}.
\end{aligned}
\end{equation}

We can think of the points in $\Gamma_{+}^{\infty}$ as painted red, and the points in $ \Gamma_{-}^{\infty}$ as painted blue. Each $(i,j)\in \N^{2}$ thus has a well-defined color. The two colors are separated through the competition interface $\phi$: Set $\phi_{0}:=(0,0)$ and define $\phi=(\phi_0, \phi_1, \phi_2,\ldots)$ inductively via
\begin{equation}\label{compint}
 \phi_{n+1}=\begin{cases}
 \phi_{n}+(1,0) \quad \mathrm{ if }\quad \phi_{n}+(1,1) \in \Gamma_{+}^{\infty},\\
 \phi_{n}+(0,1) \quad \mathrm{ if }\quad \phi_{n}+(1,1) \in \Gamma_{-}^{\infty},
 \end{cases}
\end{equation}
and we write $\phi_n=(I_n, J_n)$. We always have $\phi_1 =(1,0)$ in our model since $L_{\mathcal{L}^{-}\to (1,1)}=\omega_{1,1}<\omega_{1,1}+\omega_{0,1}\leq L_{\mathcal{L}^{+}\to (1,1)}$. Note that $I_n+J_n = n$ and $(k,n-k)$ is red for $0\leq k < I_n$ and blue for $I_n <k \leq n$.

In~\cite{PFJBLP09} some aspects of the competition interface are studied. In Theorem~1 of~\cite{PFJBLP09} it is shown (under the assumption that $\mathcal{L}^\pm$ have asymptotically fixed directions) that $\phi_n/|\phi_n|\to (\cos(\theta),\sin(\theta))$ almost surely, for some $\theta$, which might be random or deterministic. In Theorem~2 of~\cite{PFJBLP09} they determined the distribution function of $\theta$ for the case corresponding to TASEP with Bernoulli-Bernoulli initial conditions with higher density on $\Z_-$ than $\Z_+$. This corresponds to a situation with a macroscopically decreasing particle density profile (the so-called rarefaction fan). Furthermore, for initial conditions generating a shock for TASEP, they prove that $\theta$ is non-random (with an explicit value given in term of the shock velocity). Thus for situation with a shock for TASEP, their result is a law of large number type of result. In our contribution we want to analyze the fluctuations with respect to the law of large number. It is a CLT type of result, with the particularity that the fluctuations live in a $\Or(n^{1/3})$ scale and the distribution function is not Gaussian.

In Theorem~\ref{Genthm} we determine the distribution function for $n\mapsto I_n-J_n$ (properly centered and scaled) under a few assumptions, which need to be verified case by case (as they depend on $\mathcal{L}^\pm$ and the chosen law of the randomness $\omega$). In order to illustrate the result, we now present one special model in which all the assumptions are verified, see Theorem~\ref{comp2speed}.
\begin{figure}
\begin{center}
\psfrag{R}[c]{Red}
\psfrag{B}[c]{Blue}
\psfrag{Lp}[c]{${\cal L}^+$}
\psfrag{Lm}[c]{${\cal L}^-$}
\psfrag{fn}[l]{$\phi_n$}
\includegraphics[height=5cm]{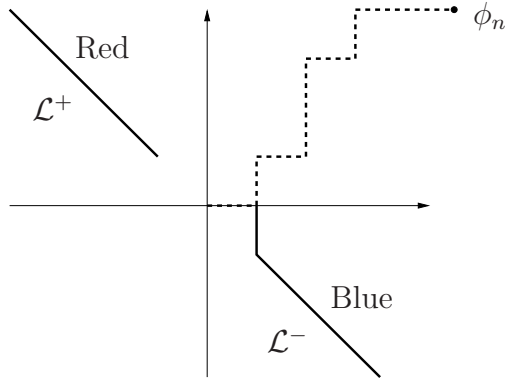}
\caption{The LPP model of Theorem~\ref{comp2speed} and the associated competition interface (dashed line).}
\label{fig1}
\end{center}
\end{figure}

Let $(\phi_{n})_{n\in {\N}}=((I_n, J_n))_{n\in \N}$ be the competition interface \eqref{compint} associated to the LPP model with $\mathcal{L}=\mathcal{L}^{+}\cup \mathcal{L}^{-}$ where $\mathcal{L}^{+} ,\mathcal{L}^{-} $ are given by \eqref{Lpm} and
\begin{equation}\label{IC}
x_{k}(0)=-2k, \, k\in \Z \setminus\{ 0\} \quad \textrm{and}\quad x_{0}(0)=1,
\end{equation}
see Figure~\ref{fig1}. The weights are
\begin{equation}\label{weights}
\omega_{i,j}:=\begin{cases}
 \exp(1) & \textrm{if } (i,j)\in \mathcal{L}^{c}, j>0,\\
 \exp(\alpha) & \textrm{if } (i,j)\in \mathcal{L}^{c}, j\leq 0,
\end{cases}
\end{equation}
where $\alpha <1$.
Setting $\eta_0=\alpha/(2-\alpha)$, $\eta=\eta_0+u\ell^{-2/3}$, and $\mu=4/(2-\alpha)$ Corollary~2.2 of~\cite{FN14} gives
\begin{equation}\label{eq2.10}
\begin{aligned}
\lim_{\ell\to \infty} \Pb(L_{\mathcal{L}^{+}\to (\eta \ell,\ell)}\leq \mu \ell+s \ell^{1/3})&=F_{\rm GOE}\left(\frac{s-2u}{\sigma_1}\right),\\
\lim_{\ell\to \infty} \Pb(L_{\mathcal{L}^{-}\to (\eta \ell,\ell)}\leq \mu \ell+s \ell^{1/3})&=F_{\rm GOE}\left(\frac{s-2u/\alpha}{\sigma_2}\right),
\end{aligned}
\end{equation}
where $F_{\GOE}$ is the density of the GOE Tracy-Widom distribution of random matrices~\cite{TW96} and $\sigma_1=\frac{2^{2/3}}{(2-\alpha)^{1/3}}$ and $\sigma_2=\frac{2^{2/3}(2-2\alpha+\alpha^2)^{1/3}}{\alpha^{2/3}(2-\alpha)}$.

\begin{thm}\label{comp2speed}
Consider the setting of (\ref{weights}). Then the fluctuations of the competition interface are asymptotically given by
\begin{equation}
\lim_{t\to \infty}\Pb\left(\frac{ I_{\lfloor t\rfloor}-J_{\lfloor t \rfloor} - (\alpha-1)t}{t^{1/3}}\leq s \right)\\
=\Pb(\chi_1-\chi_2\geq 0),
\end{equation}
where $\chi_1$ and $\chi_2$ are independent random variables with distributions
\begin{equation}
\Pb(\chi_1\leq x)=F_{\rm GOE}\left(\frac{x-\gamma s}{\sigma_1}\right),\quad \Pb(\chi_2\leq x)=F_{\rm GOE}\left(\frac{x-\gamma s/\alpha}{\sigma_2}\right),
\end{equation}
with $\gamma=2^{4/3}/(2-\alpha)^{4/3}$.
\end{thm}

In short, to prove Theorem~\ref{comp2speed}, one first observes that the event \mbox{$\{I_{\lfloor t\rfloor}-J_{\lfloor t \rfloor} \leq -t+ 2M\}$}, for $0\leq M\leq t$ amounts to the event that $(M, t-M)$ is blue, i.e., $L_{\mathcal{L}^{-}\to (M,t-M)}- L_{\mathcal{L}^{+}\to (M,t-M)} >0$ (see Proposition~\ref{translat}). Thus we need to choose $M=\alpha t/2+s t^{1/3}/2$. Secondly, the key property one needs to show is that the (properly rescaled) random variables $L_{\mathcal{L}^{+}\to (M,t-M)}$ and $L_{\mathcal{L}^{-}\to (M,t-M)}$ are asymptotically independent. Together with their limiting distributions (\ref{eq2.10}) leads to Theorem~\ref{comp2speed}, where to fit the parameters we need to set $\ell=(2-\alpha)t/2-s t^{1/3}/2$ and $u=s 2^{1/3}/(2-\alpha)^{4/3}$.

\subsubsection{Competition Interfaces and second class particles}\label{SectSecondClass}
One of the motivations to study the competition interface \eqref{compint} is its direct connection to second class particles established in~\cite{PFLP05} and~\cite{PFJBLP09}. To define the second class particle one starts TASEP in \mbox{$\eta, \eta^{\prime}\in \{0,1\}^{\mathbb{Z}}$} where \mbox{$\eta, \eta^{\prime}$} differ only at the origin, and couples the two processes in the basic coupling (in this coupling, particles from $\eta, \eta^{\prime}$ use the same Poisson processes for their jumps, the graphical construction of TASEP behind this goes back to Harris~\cite{Har78}, see also~\cite{PaFer15} for an explanation). Under this coupling, the two processes differ at one site for all time denoted by
\begin{equation}
X(t)= \sum_{x \in \Z}x \mathbbm{1}_{\{\eta_t (x) \neq \eta^{\prime}_t (x)\}}.
\end{equation}
$X(t)$ is the position at time $t$ of the second class particle which started at the origin. To any $\eta \in \{0,1\}^{\Z}$ we can associate the empirical measure
\begin{equation}
\pi^{n}(\eta)=\frac{1}{n}\sum_{i\in \Z} \eta(i)\delta_{i/n}.
\end{equation}

The initial data $\eta_0$ of a TASEP can either be deterministic (as in \eqref{IC}) or random as in Section~\ref{SectBernoulli} (the initial data is always independent of the evolution of the process). Furthermore one can also have a sequence of initial data $(\eta_{0}^{n}, n\geq 1)$. Let $(\eta_{0}^{n},n \geq 1)$ be defined on some probability space with measure $\mathbb{P}_0$. Then, the initial particle density is given by a measurable function $\rho_0 (x), x \in \R,$ if for all continuous, compactly supported $f$ on $\R$
\begin{equation}
 \lim_{n\to \infty}\int_\R \dx \pi^{n}(\eta_{0}^{n}) f(x)= \int_\R \dx x \rho_{0} (x) f(x) \quad \mathbb{P}_{0}\,\mathrm{a.s.}
\end{equation}
For example, the $\rho_0$ corresponding to \eqref{IC} is simply constant $1/2$. The density of TASEP in the large time limit is then the entropy solution to the Burgers equation~\cite{Ev10} with initial data given by $\rho_0,$ see e.g.~\cite{PaFer15} and references therein.
An important special case
are initial particle densities $\rho_0$ of the form
\begin{equation} \label{Riemann}
\rho_0 (\xi)=\begin{cases}
\rho_-& \textrm{if } \xi <0, \\
\rho_+& \textrm{if } \xi \geq 0.
\end{cases}
\end{equation}
If now $\rho_- <\rho_+$, the solution of the Burgers equation with initial data \eqref{Riemann} which gives the $t\to \infty$ density profile of TASEP
is given by
\begin{equation} \label{Riemannburg}
\rho(\xi,1)=\begin{cases}
\rho_- &\textrm{if }\xi <1-\rho_- -\rho_+, \\
\rho_+ &\textrm{if }\xi \geq 1-\rho_- -\rho_+.
\end{cases}
\end{equation}
The discontinuity in \eqref{Riemannburg} is called shock and $X(t)/t$ converges a.s.\ to $1-\rho_- -\rho_+$. Thus $X(t)$ provides an interpretation of the random shock position; for Riemann initial data with $\rho_- > \rho_+$ created by Bernoulli initial data, $X(t)/t$ is asymptotically uniformly distributed among the characteristics emanating from the origin, see e.g.\ Theorem 3 in~\cite{PFJBLP09} for both results. In the construction given in Proposition 2.2 in~\cite{PFJBLP09}, an initial configuration with a second class particle at the origin becomes an initial configuration with a hole at 0 and a particle at 1 (note that we only defined the competition interface for such initial data). Associating to such initial data the LPP model \eqref{LPPTASEP} and the competition interface \eqref{compint}, $I_n -J_n -1$ becomes then the position of the second class particle after its $(n-1)\mathrm{th}$ jump ($n\geq 1$). More precisely, with $\tau_n = L_{\mathcal{L}\to \phi_n}$ we define for $t\geq 0$
\begin{equation}\label{eq2.13}
(I(t),J(t))= \phi_n \quad \mathrm{if} \, t \in [\tau_n, \tau_{n+1}).
\end{equation}
Then, in the coupling of~\cite{PFJBLP09} $I(t)-1$ (resp.\ $J(t)$) equals the number of rightward (resp.\ leftward) jumps of the second class particle in $[0,t]$, note $\tau_1 =0$ in our setting\footnote{In~\cite{PFJBLP09}, the choice of $\mathcal{L}$ and $\omega_{i,j}$ differs slightly from the one in \eqref{LPPTASEP}, such that \mbox{$(I(t)-J(t))_{t\geq 0}$} equals $(X(t))_{t\geq 0}$ in~\cite{PFJBLP09}.}.

\subsection{General Model and Theorem}\label{gensect}
In this section we consider the general LPP model defined by \eqref{startweights}, \eqref{startLPP} \eqref{Lpm}.
We prove a generalization of Theorem~\ref{comp2speed} about the convergence of $I_n -J_n$ under three assumptions, see Theorem~\ref{Genthm} below. As noted after Theorem~\ref{comp2speed}, one looks at the probability that $L_{\mathcal{L}^{-}\to (M,t-M)}- L_{\mathcal{L}^{+}\to (M,t-M)}$ is positive for suitable $M=M(t)$. Hence, $L_{\mathcal{L}^{-}\to (M,t-M)}- L_{\mathcal{L}^{+}\to (M,t-M)}$ need to converge under the same rescaling (otherwise the probability of their difference being positive converges to $0$ or $1$). This is the content of Assumption~1. Furthermore, one needs to establish asymptotic independence of the two LPP times.
If the maximizing paths $\pi^{\max}_{+}$ of $L_{\mathcal{L}^{-}\to (M,t-M)}$ and $\pi^{\max}_{-}$ of $L_{\mathcal{L}^{+}\to (M,t-M)}$ uses different (non-random) sets of the LPP random waiting times with high probability as $t\to\infty$, then the two LPP are asymptotically independent. Close to the endpoint, this can not possibly hold, since typically $\pi^{\max}_{+}$ and $\pi^{\max}_{-}$ will actually intersect. However, if the fluctuations of the two LPP are on the leading order not affected by the randomness in a region of radius $t^\nu$, $\nu<1$ away from $(M,t-M)$ (Assumption~2), then the maximizers until that distance are supported with high probability on disjoint set of points (Assumption~3), and the two LPP are asymptotically independent. The framework of these Assumptions was used in~\cite{FN14} to prove the limit law of LPP times. Here we prove convergence of $I_n -J_n$ and obtain Theorem~\ref{comp2speed} as a corollary.

In the following, we do not always write the integer parts, since even any perturbation of order $o(t^{1/3})$ becomes irrelevant in the limit, see \eqref{keypiece} below. For convenience, we assume that all appearing distribution functions are continuous, otherwise statements will only hold at continuity points.

\begin{assumption}\label{Assumpt1}
Fix $\eta_0\in (0,\infty)$ and $\eta=\eta_0+u t^{-2/3}$. Assume that there exists some $\mu$ such that
\begin{equation}\label{eq4a}
\lim_{ t\to\infty} \Pb\left(\frac{L_{\mathcal{L}^+\to (\eta t, t)}-\mu t}{ t^{1/3}}\leq s\right) = G_1(s;u),
\end{equation}
and
\begin{equation}\label{eq4b}
\lim_{ t\to\infty} \Pb\left(\frac{L_{\mathcal{L}^-\to (\eta t, t)}-\mu t}{ t^{1/3}}\leq s\right) = G_2(s;u),
\end{equation}
where $G_1$ and $G_2$ are some (continuous) distribution functions depending on $u$.
\end{assumption}

Secondly, we assume that there is a point $E^+$ at distance of order $ t^\nu$, for some $1/3<\nu<1$, which lies on the characteristic from $\mathcal{L}^+$ to $E=(\eta t,t)$ and that there is slow-decorrelation as in Theorem~2.1 of~\cite{CFP10b}.
\begin{assumption}\label{Assumpt2}
Fix $\eta_0\in (0,\infty)$ and $\eta=\eta_0+u t^{-2/3}$. Assume that we have a point $E^+=(\eta t-\kappa t^\nu, t- t^\nu)$ such that for some $\mu_0$, and $\nu\in (1/3,1)$ it holds
\begin{equation}\label{eq5}
\lim_{ t\to\infty} \Pb\left(\frac{L_{E^+\to (\eta t, t)}-\mu_0 t^\nu}{ t^{\nu/3}}\leq s\right) = G_0(s;u),
\end{equation}
and
\begin{equation}\label{eq6}
\lim_{ t\to\infty} \Pb\left(\frac{L_{\mathcal{L}^+\to E^+}-\mu t+\mu_0 t^\nu}{ t^{1/3}}\leq s\right) = G_1(s;u),
\end{equation}
where $G_0$ and $G_1$ are (continuous) distribution functions.
\end{assumption}

\begin{assumption}\label{Assumpt3}
Let $\nu$ be as in Assumption~\ref{Assumpt2}. Consider the points \mbox{$D_{\gamma}=(\lfloor \gamma \eta t \rfloor,\lfloor \gamma t\rfloor)$} with \mbox{$\gamma \in [0,1- t^{\beta-1}]$} and $\beta \in (0,\nu)$. Assume that
\begin{equation}
\begin{aligned}
\lim_{ t\to\infty}\Pb\bigg(\bigcup_{D_{\gamma}\atop \gamma \in [0,1- t^{\beta-1}]}\left\{D_\gamma\in \pi^{\rm max}_{L_{\mathcal{L}^+\to E^+}}\right\}\bigg) &=0,\\
\lim_{ t\to\infty}\Pb\bigg(\bigcup_{D_{\gamma}\atop \gamma \in [0,1- t^{\beta-1}]}\left\{D_\gamma\in \pi^{\rm max}_{L_{\mathcal{L}^-\to (\eta t, t)}}\right\}\bigg) &=0.
\end{aligned}
\end{equation}
\end{assumption}

Then, under the proceeding Assumptions, we have the following Theorem.
\begin{thm}\label{Genthm}
Assume Assumptions 1, 2 and 3 hold with
\begin{equation}
\eta=\eta_0 + ut^{-2/3}, \quad u \in \R.
\end{equation}
Then, for any sequence $a_t =o(t^{1/3})$ we have
\begin{equation}
\lim_{t\to\infty}\Pb\left( I_{\lfloor t\rfloor}-J_{\lfloor t\rfloor}\leq -t \frac{1-\eta_0 }{1+\eta_0 }+ \frac{2u}{(1+\eta_0 )^{4/3}} t^{1/3} +a_t\right)=\Pb_{G_{2}\star G_{1,-}}((0,\infty)),
\end{equation}
where $\Pb_{G_{2}\star G_{1,-}}$ is the convolution of the probability measures induced by the distribution functions $G_{2}(x;u)$ and $G_{1,-}(x;u):=1-G_{1}(-x;u)$.
\end{thm}
\begin{proof}[Proof of Theorem~\ref{comp2speed}]
In the proof of Corollary 2.2 of~\cite{FN14}, it has been shown that the Assumptions 1, 2 and 3 hold ($x_{0}(0)=1$ rather than $x_{0}(0)=0$ clearly do not affect the asymptotic behavior. This is easily seen by comparison and basic coupling, see remark at page 7 of~\cite{BFS07}) with $\eta_0 = \frac{\alpha}{2-\alpha}$ and
\begin{equation}
\begin{aligned}
G_{1}(x;u) =F_{\GOE}\left(\frac{x -2u }{\sigma_1} \right), \quad
G_{2}(x;u) =F_{\GOE}\left( \frac{x -2u/\alpha }{\sigma_2}\right),
\end{aligned}
\end{equation}
where $\sigma_1, \sigma_2$ were defined in Theorem~\ref{comp2speed}. Then, taking $2u=2^{4/3}/(2-\alpha)^{4/3}s$ gives
\begin{equation}
 -t \frac{1-\eta_0 }{1+\eta_0 }+ \frac{2u}{(1+\eta_0 )^{4/3}} t^{1/3}=(\alpha-1)t +s t^{1/3}+ \mathcal{O}(t^{-1/3})
\end{equation}
from which the result follows using Theorem~\ref{Genthm}.
\end{proof}

\subsection{Proof of Theorem~\ref{Genthm}}
In the following, we will often use the following elementary lemma from~\cite{BC09}. By $``\Rightarrow"$ we designate convergence in distribution.
\begin{lem}[Lemma 4.1 in~\cite{BC09}]\label{lemma4.1}
Let $D$ be a probability distribution and $(X_n)_{n\in \mathbb{N}}, (\tilde{X}_n)_{n\in \mathbb{N}}$ be sequences of random variables.
If $X_n\geq \tilde{X}_n$ and $\tilde{X}_n, X_n \Rightarrow D,$ then $X_n-\tilde{X}_n$ converges to zero in probability.
If $X_n \Rightarrow D$ and $X_n-\tilde{X}_n$ converges to zero in probability, then $\tilde{X}_n \Rightarrow D$ as well.
\end{lem}

We denote by
\begin{equation}
L_{\mathcal{L}^{+}\to (\eta t,t)}^{\mathrm{resc}}=\frac{L_{\mathcal{L}^+\to (\eta t, t)}-\mu t}{ t^{1/3}}
\end{equation}
the last passage time $L_{\mathcal{L}^{+}\to (\eta t,t)}$ rescaled as required by Assumption~\ref{Assumpt1}. We define analogously $L_{\mathcal{L}^{-}\to (\eta t,t)}^{\mathrm{resc}}, L_{E^{+}\to (\eta t,t)}^{\mathrm{resc}}$ and $L_{\mathcal{L}^{+}\to E^{+}}^{\mathrm{resc}}$ as the last passage times rescaled as required by Assumption~\ref{Assumpt1} resp.~\ref{Assumpt2}.

To study the asymptotic behavior of $I_n -J_n$ we make the following observation.
\begin{prop}\label{translat}
Let $\phi$ be the competition interface of the LPP model \eqref{startLPP}, \eqref{startweights} with $x_{0} (0)=1, x_{1} (0)<-1$.
Let $n,M\in \N$ and $M\leq n -1$. Then
\begin{equation}
\begin{aligned}
\Pb\left((M,n-M)\in \Gamma_{-}^{\infty}\right)&\leq \Pb\left(I_n-J_n \leq -n+2M\right)
\\&\leq \Pb\left((M,n-M)+(1,-1)\in \Gamma_{-}^{\infty}\right).
\end{aligned}
\end{equation}
\end{prop}
\begin{proof}
Note that since $I_n+J_n = n$ we have the equality of the events
\begin{equation}
\{I_n -J_n\leq -n+2M\}=\{\phi_n \in\{(k,n-k),0\leq k\leq M\}\},
\end{equation}
with $\phi_n$ given in (\ref{eq2.13}). The statement follows now from
\begin{equation}
\begin{aligned}
\{(M,n-M)\in \Gamma_{-}^{\infty} \}&\subseteq \{\phi_n \in \{(k,n-k),0\leq k\leq M\}\}\\& \subseteq\{(M,n-M)+(1,-1)\in \Gamma_{-}^{\infty} \}.
\end{aligned}
\end{equation}
\end{proof}
With this observation, we can translate the behavior of $I_n -J_n$ into the difference of LPP times.
\begin{prop}\label{prop2}
Let $\eta=\eta_0 + ut^{-2/3}$, with $u \in \R$ and suppose that
\begin{equation}
L_{\mathcal{L}^{-}\to (\eta t,t)}^{\mathrm{resc}}-L_{\mathcal{L}^{+}\to (\eta t,t)}^{\mathrm{resc}} \Rightarrow D\textrm{ as }t\to\infty.
\end{equation}
Then, under Assumption~1, for any sequence $a_t =o(t^{1/3})$ we have
\begin{equation}
\lim_{t\to\infty}\Pb\left( I_{\lfloor t\rfloor}-J_{\lfloor t\rfloor}\leq -t \frac{1-\eta_0 }{1+\eta_0 }+ \frac{2u}{(1+\eta_0 )^{4/3}} t^{1/3} +a_t \right)
=\Pb_{D}((0,\infty)),
\end{equation}
where $\Pb_{D}$ is the probability measure with distribution $D$.
\end{prop}
\begin{proof}
Let us define
\begin{equation}
\hat{\eta}=\eta_0 + u(t/(1+\eta_0))^{-2/3}, \quad n(t)=\lfloor t \rfloor, \quad M(t)=\left\lfloor \frac{\hat{\eta}t}{1+\hat{\eta}} \right\rfloor.
\end{equation}
Then setting $\ell=\frac{t}{1+\hat{\eta}}$ we have with $\eta(\ell)=\eta_0+ u \ell^{-2/3}$ and
\begin{equation}
(M(t),n(t)-M(t))=(\eta(\ell) \ell, \ell)+c_\ell,
\end{equation}
with $c_\ell=(c_\ell^{1},c_\ell^{2})=o(\ell^{1/3})$.

What we have to show is that if $L_{\mathcal{L}^{-}\to (\eta (\ell)\ell,\ell)}^{\mathrm{resc}}-L_{\mathcal{L}^{+}\to (\eta(\ell) \ell,\ell)}^{\mathrm{resc}} \Rightarrow D,$ then also for any $b_\ell=(b_\ell^{1},b_\ell^{2})=o(\ell^{1/3})$
\begin{equation}
\label{keypiece}
L_{\mathcal{L}^{-}\to (\eta(\ell)\ell ,\ell )+b_\ell}^{\mathrm{resc}}-L_{\mathcal{L}^{+}\to (\eta(\ell)\ell ,\ell )+b_\ell}^{\mathrm{resc}} \Rightarrow D.
\end{equation}
Indeed, given \eqref{keypiece}, it follows from Proposition~\ref{translat} that
\begin{equation}
\lim_{t\to\infty}\Pb\left( I_{\lfloor t\rfloor}-J_{\lfloor t\rfloor}\leq -n(t)+2M(t)\right)
=\Pb_{D}((0,\infty)).
\end{equation}
Furthermore, if $a_t =o(t^{1/3})$ there is an integer $\tilde{M}(t)$ such that
\begin{equation}
-n(t)+2M(t)+\lfloor a_t \rfloor =-n(t)+2\tilde{M}(t), \quad M(t)-\tilde{M}(t)=o(t^{1/3}).
\end{equation} Applying Proposition~\ref{translat} with $n(t), \tilde{M}(t)$ and then using \eqref{keypiece} we obtain
\begin{equation}\label{integpart}
\lim_{t\to\infty}\Pb\left( I_{\lfloor t\rfloor}-J_{\lfloor t\rfloor}\leq -n(t)+2M(t) +\lfloor a_t \rfloor\right)
=\Pb_{D}((0,\infty)),
\end{equation}
which is the statement to be proven. So let us now prove \eqref{keypiece}.
Writing
\begin{equation}
\begin{aligned}
X_{\ell }&=L_{\mathcal{L}^{-}\to (\eta(\ell)\ell ,\ell )}^{\mathrm{resc}}, \quad Y_\ell =L_{\mathcal{L}^{+}\to (\eta(\ell)\ell,\ell )}^{\mathrm{resc}},\\
Z_\ell &= \frac{ L_{\mathcal{L}^{-}\to (\eta(\ell)\ell ,\ell )+b_\ell}-L_{\mathcal{L}^{-}\to (\eta(\ell)\ell ,\ell )}}{\ell^{1/3}}+\frac{L_{\mathcal{L}^{+}\to (\eta(\ell)\ell ,\ell )}-L_{\mathcal{L}^{+}\to (\eta(\ell)\ell ,\ell )+b_\ell}}{\ell^{1/3}},
\end{aligned}
\end{equation}
we have
\begin{equation}L_{\mathcal{L}^{-}\to (\eta(\ell)\ell ,\ell )+b_\ell}^{\mathrm{resc}}-L_{\mathcal{L}^{+}\to (\eta(\ell)\ell ,\ell )+b_\ell}^{\mathrm{resc}} =X_\ell- Y_\ell + Z_\ell,
\end{equation}
so it suffices to show $Z_\ell \Rightarrow 0$. Let $L=\max \left\{ \left|\frac{b^{1}_\ell}{\eta}\right|, \left|b_{\ell}^{2}\right|\right\}$.
Then
\begin{equation}\label{ineq}
\left|L_{\mathcal{L}^{-}\to (\eta(\ell)\ell ,\ell )+b_\ell}^{\mathrm{resc}}-L_{\mathcal{L}^{-}\to (\eta(\ell)\ell ,\ell )}^{\mathrm{resc}}\right|\leq \left|\widehat L_{\mathcal{L}^{-}\to (\eta(\ell) (\ell-L) ,\ell -L)}^{\mathrm{resc}}-\widehat L_{\mathcal{L}^{-}\to (\eta(\ell) (\ell +L) ,\ell +L)}^{\mathrm{resc}}\right|
\end{equation}
where we defined
\begin{equation}
\widehat L_{\mathcal{L}^{-}\to (\eta(\ell) (\ell \pm L) ,\ell \pm L)}^{\mathrm{resc}}=\frac{L_{\mathcal{L}^{-}\to (\eta(\ell) (\ell \pm L) ,\ell \pm L)}-\mu \ell}{\ell^{1/3}}.
\end{equation}
Now
\begin{equation}
\widehat L_{\mathcal{L}^{-}\to (\eta(\ell) (\ell-L) ,\ell -L)}^{\mathrm{resc}}\leq \widehat L_{\mathcal{L}^{-}\to (\eta(\ell)( \ell +L) ,\ell +L)}^{\mathrm{resc}},
\end{equation}
and $L_{\mathcal{L}^{-}\to (\eta(\ell) (\ell-L) ,\ell -L)}^{\mathrm{resc}},L_{\mathcal{L}^{-}\to (\eta(\ell) (\ell+L) ,\ell +L)}^{\mathrm{resc}}$ both converge to the same distribution by Assumption~1. Indeed, setting $\ell_\pm=\ell \pm L$ we get
\begin{equation}\label{genau}
\begin{aligned}
\widehat L_{\mathcal{L}^{-}\to (\eta(\ell) (\ell \pm L) ,\ell \pm L)}^{\mathrm{resc}}&=\frac{L_{\mathcal{L}^{-}\to (\eta (\ell_\pm)\ell_\pm ,\ell_\pm)}-\mu \ell_\pm }{\ell_\pm^{1/3}}+\mathcal{O}\left(L \ell^{-1/3}\right)\\&+\mathcal{O}(L \ell_\pm^{-4/3})(L_{\mathcal{L}^{-}\to (\eta (\ell) \ell_\pm, \ell_\pm)}-\mu \ell_\pm)\\&+\frac{L_{\mathcal{L}^{-}\to (\eta (\ell)\ell_\pm ,\ell_\pm)}-L_{\mathcal{L}^{-}\to (\eta (\ell_\pm)\ell_\pm ,\ell_\pm)}}{\ell_\pm^{1/3}}
\end{aligned}
\end{equation}
and all terms except the first one on the right-hand side of \eqref{genau} converge to zero in probability.
 Hence by Lemma~\ref{lemma4.1}
$\widehat L_{\mathcal{L}^{-}\to (\eta(\ell) (\ell-L) ,\ell -L)}^{\mathrm{resc}}- \widehat L_{\mathcal{L}^{-}\to (\eta(\ell) (\ell+L) ,\ell +L)}^{\mathrm{resc}}$
 converges to zero in probability and so does the left-hand side of \eqref{ineq}. An analogous argument shows that
$\frac{L_{\mathcal{L}^{+}\to (\eta(\ell)\ell ,\ell )}-L_{\mathcal{L}^{+}\to (\eta(\ell)\ell ,\ell )+b_\ell}}{\ell^{1/3}}$ converges to zero in probability. This implies that $Z_\ell \Rightarrow 0$.
\end{proof}

In view of the previous proposition, we need to determine the limit law of $L_{\mathcal{L}^{-}\to (\eta t,t)}^{\mathrm{resc}}-L_{\mathcal{L}^{+}\to (\eta t,t)}^{\mathrm{resc}} $. For this, we proceed as in~\cite{FN14}, reducing the problem to the limit law of two independent random variables in the following three Propositions.
\begin{prop}\label{red1}
Suppose Assumptions 1 and 2 hold, and let $D$ be a probability distribution. If
\begin{equation}
L_{\mathcal{L}^{-}\to (\eta t,t)}^{\mathrm{resc}}-\frac{L_{\mathcal{L}^{+}\to E^{+}}+L_{E^{+}\to (\eta t,t)}-\mu t}{t^{1/3}} \Rightarrow D\textrm{ as }t\to\infty,
\end{equation}
then
\begin{equation}
L_{\mathcal{L}^{-}\to (\eta t,t)}^{\mathrm{resc}}-L_{\mathcal{L}^{+}\to (\eta t,t)}^{\mathrm{resc}} \Rightarrow D\textrm{ as }t\to\infty.
\end{equation}
\end{prop}
\begin{proof}
We set
\begin{equation}
\begin{aligned}
X_n &= L_{\mathcal{L}^{-}\to (\eta t,t)}^{\mathrm{resc}}-\frac{L_{\mathcal{L}^{+}\to E^{+}}+L_{E^{+}\to (\eta t,t)}-\mu t}{t^{1/3}},\\
\tilde{X}_n &= L_{\mathcal{L}^{-}\to (\eta t,t)}^{\mathrm{resc}}-L_{\mathcal{L}^{+}\to (\eta t,t)}^{\mathrm{resc}}.
\end{aligned}
\end{equation}
By Lemma~\ref{lemma4.1} it suffices to show $X_n -\tilde{X}_n $ converges to $0$ in probability. Now
\begin{equation}
X_n - \tilde{X}_n =L_{\mathcal{L}^{+}\to (\eta t,t)}^{\mathrm{resc}}-\frac{L_{\mathcal{L}^{+}\to E^{+}}+L_{E^{+}\to (\eta t,t)}-\mu t}{t^{1/3}},
\end{equation}
and since
\begin{equation}
L_{\mathcal{L}^{+}\to (\eta t,t)}^{\mathrm{resc}}\geq \frac{L_{\mathcal{L}^{+}\to E^{+}}+L_{E^{+}\to (\eta t,t)}-\mu t}{t^{1/3}},
\end{equation}
to show that $X_n - \tilde{X}_n$ converges to zero in probability, it suffices again by Lemma~\ref{lemma4.1} to show that $L_{\mathcal{L}^{+}\to (\eta t,t)}^{\mathrm{resc}}$ and $\frac{L_{\mathcal{L}^{+}\to E^{+}}+L_{E^{+}\to (\eta t,t)}-\mu t}{t^{1/3}}$ converge to the same distribution as $t\to\infty$. By Assumption~1, in the large-$t$ limit, $\Pb(L_{\mathcal{L}^{+}\to (\eta t,t)}^{\mathrm{resc}}\leq s)\to G_{1}(s)$ and by Assumption~2 $\Pb(L_{\mathcal{L}^{+}\to E^{+}}^{\mathrm{resc}}\leq s)\to G_{1}(s)$. Furthermore, by Assumption~2 for any $\varepsilon>0$
\begin{equation}\label{E+0}
\Pb\left(\left|\frac{L_{E^{+}\to (\eta t,t)}-\mu_{0}t^{\nu}}{t^{\nu/3}} \right|\geq \varepsilon t^{(1-\nu)/3}\right) \to 0\textrm{ as }t\to\infty.
\end{equation}
Thus $\frac{L_{E^{+}\to (\eta t,t)}-\mu_0 t^{\nu}}{t^{1/3}}$ converges to zero in probability as $t\to\infty$. Consequently
\begin{equation}
\begin{aligned}
\Pb\left(\frac{L_{\mathcal{L}^{+}\to E^{+}}+L_{E^{+}\to (\eta t,t)}-\mu t}{t^{1/3}}\leq s\right)&=\Pb \left( L_{\mathcal{L}^{+}\to E^{+}}^{\mathrm{resc}}+\frac{L_{E^{+}\to (\eta t,t)}-\mu_{0}t^{\nu}}{t^{1/3}}\leq s\right)
\\& \to G_{1}(s) \textrm{ as }t\to\infty.
\end{aligned}
\end{equation}
\end{proof}

For the next reduction, define for some set $B$ and point $C$ $\tilde{L}_{B\to C}$ to be the last passage time from $B$ to $C$ except that we take the maximum only over paths not containing any point $\bigcup_{\gamma \in [0,1-t^{\beta-1}]}D_{\gamma} $ with \mbox{$D_\gamma$} as in Assumption~\ref{Assumpt3}. We rescale $\tilde{L}_{B\to C}$ just as $L_{B\to C}$.
\begin{prop}\label{red2}
Suppose that Assumptions 1, 2 and 3 hold. If
\begin{equation}\label{Y1s}
Y_{t,1}:=\tilde{L}_{\mathcal{L}^{-}\to (\eta t,t)}^{\mathrm{resc}}-\frac{\tilde{L}_{\mathcal{L}^{+}\to E^{+}}+L_{E^{+}\to (\eta t,t)}-\mu t}{t^{1/3}}\Rightarrow D \textrm{ as }t\to\infty,
\end{equation}
then
\begin{equation}\label{Y2s}
Y_{t,2}:=L_{\mathcal{L}^{-}\to (\eta t,t)}^{\mathrm{resc}}-\frac{L_{\mathcal{L}^{+}\to E^{+}}+L_{E^{+}\to (\eta t,t)}-\mu t}{t^{1/3}}\Rightarrow D \textrm{ as }t\to\infty.
\end{equation}
\end{prop}
\begin{proof}
For $\varepsilon >0$ we have
\begin{equation}
\begin{aligned}\label{D1x}
\Pb (\left|Y_{t,1}- Y_{t,2}\right|\geq \varepsilon)&\leq \Pb\left(
\bigg|\frac{L_{\mathcal{L}^{+}\to E^{+}}-\tilde{L}_{\mathcal{L}^{+}\to E^{+}} }{t^{1/3}}\bigg|\geq \varepsilon/2\right)
\\&+ \Pb\left(
\bigg|\frac{L_{\mathcal{L}^{-}\to (\eta t,t)}-\tilde{L}_{\mathcal{L}^{-}\to (\eta t,t)} }{t^{1/3}}\bigg|\geq \varepsilon/2\right).
\end{aligned}
\end{equation}
On the other hand,
\begin{equation}\label{eq2.52}
\{L_{\mathcal{L}^{+}\to E^{+}}-\tilde{L}_{\mathcal{L}^{+}\to E^{+}}=0\}\subset \bigcup_{D_{\gamma}\atop \gamma \in [0,1- t^{\beta-1}]}\left\{D_\gamma\in \pi^{\rm max}_{L_{\mathcal{L}^+\to E^+}}\right\}
\end{equation}
because if the maximizers do not reach any points $D_\gamma$, then the two last passage time are identical. The same holds for $E^+$ replaced with $(\eta t,t)$. However, by Assumption~3, the probability of the event in the r.h.s.\ of (\ref{eq2.52}) goes to $0$ as $t\to\infty.$
 Hence by Lemma~\ref{lemma4.1} if $Y_{t,1}\Rightarrow D$, also $Y_{t,2}\Rightarrow D$.
\end{proof}

Finally, we have the following.
\begin{prop}\label{red3}
If we have that
\begin{equation}\label{conv}
\tilde{L}_{\mathcal{L}^{-}\to (\eta t,t)}^{\mathrm{resc}}-\tilde{L}_{\mathcal{L}^{+}\to E^{+}}^{\mathrm{resc}}\Rightarrow D \textrm{ as }t\to\infty,
\end{equation}
 then also
\begin{equation}
\tilde{L}_{\mathcal{L}^{-}\to (\eta t,t)}^{\mathrm{resc}}-\frac{\tilde{L}_{\mathcal{L}^{+}\to E^{+}}+L_{E^{+}\to (\eta t,t)}-\mu t}{t^{1/3}}\Rightarrow D \textrm{ as }t\to\infty.
\end{equation}
\end{prop}
\begin{proof}
Simply note that $X_t =\frac{L_{E^{+}\to (\eta t,t)}-\mu_0 t^{\nu}}{t^{1/3}}$ converges to zero in probability by \eqref{E+0}, hence if \eqref{conv} holds we have by Lemma~\ref{lemma4.1}
\begin{equation}
\tilde{L}_{\mathcal{L}^{-}\to (\eta t,t)}^{\mathrm{resc}}-(\tilde{L}_{\mathcal{L}^{+}\to E^{+}}^{\mathrm{resc}}+X_t ) = \tilde{L}_{\mathcal{L}^{-}\to (\eta t,t)}^{\mathrm{resc}}-\frac{\tilde{L}_{\mathcal{L}^{+}\to E^{+}}+L_{E^{+}\to (\eta t,t)}-\mu t}{t^{1/3}}
\Rightarrow D.
\end{equation}
\end{proof}

\begin{proof}[Proof of Theorem~\ref{Genthm}]
By construction $\tilde{L}_{\mathcal{L}^{-}\to (\eta t,t)}^{\mathrm{resc}},\tilde{L}_{\mathcal{L}^{+}\to E^{+}}^{\mathrm{resc}}$ are independent, since they depend on disjoint sets of $\omega_{i,j}$'s, which are independent random variables. Furthermore, as shown in the proof of Proposition~\ref{red2},
\mbox{$\tilde{L}_{\mathcal{L}^{-}\to (\eta t,t)}^{\mathrm{resc}}-L_{\mathcal{L}^{-}\to (\eta t,t)}^{\mathrm{resc}}$} and $\tilde{L}_{\mathcal{L}^{+}\to E^{+}}^{\mathrm{resc}}
-L_{\mathcal{L}^{+}\to E^{+}}^{\mathrm{resc}}$ converge to zero in probability and hence by Assumptions 1,2 and Lemma~\ref{lemma4.1}
\begin{equation} \label{conv1}
\Pb(\tilde{L}_{\mathcal{L}^{-}\to (\eta t,t)}^{\mathrm{resc}}\leq s)\to G_{2}(s)
\quad \textrm{and}\quad\Pb(\tilde{L}_{\mathcal{L}^{+}\to E^{+}}^{\mathrm{resc}}\leq s)\to G_{1}(s)
\end{equation}
as $t\to\infty$. Now \eqref{conv1} and independence imply that the vector $
(\tilde{L}_{\mathcal{L}^{-}\to (\eta t,t)}^{\mathrm{resc}}, \tilde{L}_{\mathcal{L}^{+}\to E^{+}}^{\mathrm{resc}})$ converges in law to the product measure $\Pb_{G_{2}} \otimes \Pb_{G_{1}}$. Consequently, by the continuous mapping theorem,
$(\tilde{L}_{\mathcal{L}^{-}\to (\eta t,t)}^{\mathrm{resc}}- \tilde{L}_{\mathcal{L}^{+}\to E^{+}}^{\mathrm{resc}})$ converges to $\Pb_{G_{2}} \star \Pb_{G_{1,-}}$. Combining this with Propositions~\ref{prop2},~\ref{red1},~\ref{red2} and~\ref{red3} leads to the proof of Theorem~\ref{Genthm}.
\end{proof}

\subsection{Application to Bernoulli initial data}\label{SectBernoulli}
In this section we explain that Theorem~\ref{Genthm} applies to random shock Bernoulli initial data as well. In this case, the asymptotic independence of LPP times underlying Theorem~\ref{Genthm} had already been established in~\cite{CFP09} and does not require the detailed control over fluctuations of maximizing paths as in Assumption~3. We review the proof of asymptotic independence given in~\cite{CFP09} and explain how to obtain the result from it. For the processes $(I(t), J(t)$ from \eqref{eq2.13}, a central limit theorem is also available, see Theorem 4 in~\cite{PFJBLP09}.

Consider TASEP $(\eta_t)_{t\geq 0}$ where the $(\eta_0 (i), i\in \Z)$ are independent Bernoulli random variables with parameter $\rho_+$ (resp.\ $\rho_-$) for $i\geq 0$ (resp.\ $i<0$). Theorem~\ref{Genthm} applies when $\rho_+ > \rho_-$, i.e., in the shock case. For $\rho_+ < \rho_-$ we have a rarefaction fan and the competition interface has a random asymptotic direction, see Theorem~2(c) of~\cite{PFJBLP09}, while for $\rho_+ = \rho_-$ the competition interface has a deterministic direction but there is no underlying independence of LPP times, see~\cite{BFP09}. Hence, Theorem~\ref{Genthm} does not apply in these cases.

Let now $\mathcal{L}$ and the $\{\omega_{i,j}\}$ be given by \eqref{startLPP}, \eqref{startweights} (all particles have jump rate $1$).
While the set $\mathcal{L}$ is random in our case, an application of Burke's Theorem~\cite{Bur56} yields that the correspondence \eqref{LPPTASEP} holds for a point-to-point problem with certain boundary weights (see~\cite{PS01}). Namely, neglecting an asymptotically irrelevant set of weights which are equal to zero (see Proposition~2.2 of~\cite{FS05a}), one can choose in \eqref{LPPTASEP} $\mathcal{L}=\{(0,0)\}$ and weights\footnote{In the cited papers~\cite{PS01,FS05a,CFP09}, the parameter put in the exponential was its mean rather than the jump rate, which are inverse of each other.}
\begin{equation}\label{berweig}
\omega_{i,j}:=\begin{cases}
0 \quad &\mathrm{if\,\,} i=j=0,\\
\exp(1) \quad & \mathrm{if\,\,} i,j>0,\\
\exp(\rho_{-}) \quad & \mathrm{if\,\,} i=0, j\geq 1,
\\ \exp(1-\rho_{+}) \quad & \mathrm{if\,\,} j=0, i\geq 1,
\end{cases}
\end{equation}
see Figure~\ref{fig2}.
\begin{figure}
\begin{center}
\psfrag{Lm}[c]{${\cal L}_-$}
\psfrag{Pm}[c]{$P_-$}
\psfrag{Lp}[c]{${\cal L}_+$}
\psfrag{Pp}[c]{$P_+$}
\psfrag{e1}[c]{$\exp(1)$}
\psfrag{ep}[c]{$\exp(1-\rho_+)$}
\psfrag{em}[r]{$\exp(\rho_-)$}
\psfrag{B}[l]{$(\eta t,t)$}
\psfrag{pmaxp}[l]{$\pi^{\max}_{{\cal L}^+\to(\eta t,t)}$}
\psfrag{pmaxm}[l]{$\pi^{\max}_{{\cal L}^-\to(\eta t,t)}$}
\includegraphics[height=5cm]{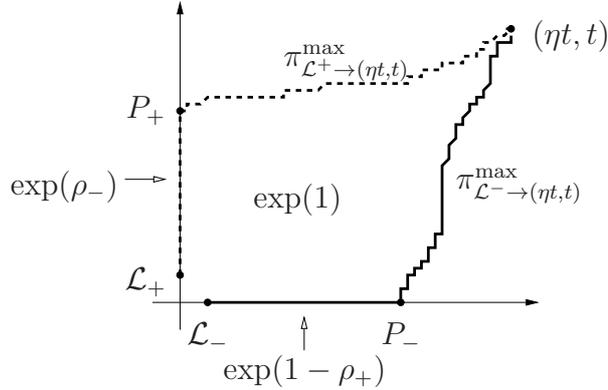}
\caption{The LPP model of Theorem~\ref{berthm} for $\rho_+ = 3/4 =1-\rho_-$.The maximizing paths $\pi^{\max}_{\mathcal{L}^{+}\to (\eta t,t)}$ and $\pi^{\max}_{\mathcal{L}^{-}\to (\eta t,t)}$ follows the boundary until points close to $P_+$ resp.\ $P_-$ after which they enter into the bulk.}
\label{fig2}
\end{center}
\end{figure}
Choosing \mbox{$\mathcal{L}^{+}=\{(0,1)\}$} and $\mathcal{L}^{-}=\{(1,0)\}$ we can again define the competition interface via $\phi_0 = (0,0)$ and \eqref{compint}.
We then obtain the following result.

\newpage
\begin{thm}\label{berthm}
Consider the LPP model defined by \eqref{berweig} with \mbox{$\mathcal{L}^{+}=\{(0,1)\}$} and $\mathcal{L}^{-}=\{(1,0)\}$. Let $(I_n, J_n)_{n\geq 1}$ be the competition interface in this model. Take $\eta = \frac{ (1-\rho_{+})(1-\rho_{-})}{ \rho_{-} \rho_+}$ and let \mbox{$v_+ = \frac{1}{\rho_{-}^{2}}-\frac{\eta}{(1-\rho_{-})^{2}}$,} $v_-=\frac{\eta}{(1-\rho_{+ })^{2}} -\frac{1 }{\rho_{+}^{2}}$ and finally $m_{+}=1/(1-\rho_{-})$,
$m_{-}=1/(1-\rho_{+})$. We then have
\begin{equation}\label{bernoullithm}
\lim_{t\to\infty}\Pb\left( I_{\lfloor t\rfloor}-J_{\lfloor t\rfloor}\leq -t \frac{1-\eta}{1+\eta }+\frac{2ut^{1/2}}{(1+\eta)^{3/2}}\right)=\int_{u(m_{+}-m_{-})}^{\infty}\dx x\frac{e^{-\frac{x^{2}}{2(v_- + v_+)}}}{\sqrt{2\pi(v_- +v_+)}}.
\end{equation}
\end{thm}
\begin{proof}
First note that Proposition~\ref{translat} is a general statement and it applies here too.
Denote now
\begin{equation}\label{resc}
L_{\mathcal{L}^{\pm}\to (\eta t +u t^{1/2},t) }^{\mathrm{resc}}=\frac{L_{\mathcal{L}^{\pm}\to (\eta t +u t^{1/2},t) } - \frac{t}{\rho_{+} \rho_{-} }}{t^{1/2}}.
\end{equation}
Furthermore, denote by $N(m,v)$ a Gaussian random variable with mean $m$ and variance $v$ and write $N_\pm =N(m_\pm u,v_{\pm})$.
We now have

\begin{align}
&\label{11}\lim_{t\to \infty}\Pb(L_{\mathcal{L}^{\pm}\to (\eta t+u t^{1/2},t) }^{\mathrm{resc}}\leq s)=\Pb(N_\pm\leq s)
\end{align}
where \eqref{11} follows for $\mathcal{L}^{+}$ from Proposition 2.8 part (b) in~\cite{CFP09} (a random matrix theory variant of this result already appeared in Theorem 1.1 of~\cite{BBP06}) by a simple change of variable, for $ L_{\mathcal{L}^{-}\to (\eta t +u t^{1/2},t)}$ one has to look at the transposed LPP model with weights $\tilde{\omega}_{i,j}=\omega_{j,i}$ and transposed endpoint
to apply Proposition 2.8 of~\cite{CFP09}.
 We can see \eqref{11} as confirming the analogue of Assumption~1 in this new setting.
In particular, Proposition~\ref{prop2} still holds, i.e., we have
\begin{equation}
\eqref{bernoullithm}=\lim_{t\to \infty}\Pb( L_{\mathcal{L}^{-}\to (\eta t +u t^{1/2},t) }^{\mathrm{resc}}-L_{\mathcal{L}^{+}\to (\eta t +u t^{1/2},t) }^{\mathrm{resc}}>0).
\end{equation}
Thus one has to establish the asymptotic independence of $L_{\mathcal{L}^{+}\to (\eta t +u t^{1/2},t)}$ and $L_{\mathcal{L}^{-}\to (\eta t +u t^{1/2},t)}$, which is done in~\cite{CFP09} using coupling arguments. Namely, one defines the points
\begin{equation}
\begin{aligned}
P_{+}&=\left(0,t-\frac{\eta t}{(1/\rho_{-} -1)^{2}}+\frac{u t^{1/2}\rho_{-}}{1-\rho_{-}}\right),\\
P_{-}&= \left(\eta t - \frac{t}{(1/(1-\rho_{+}) -1)^{2}} +u t^{1/2},0\right),
\end{aligned}
\end{equation}
and the random variables
\begin{equation}
\begin{aligned}
 &X_{\mathcal{L}^{+}\to (\eta t +u t^{1/2},t)}=L_{\mathcal{L}^{+}\to P_{+}} + L_{P_{+}\to (\eta t +u t^{1/2},t)},\\
 & X_{\mathcal{L}^{-}\to (\eta t +u t^{1/2},t)}=L_{\mathcal{L}^{-}\to P_{-}} + L_{P_{-}\to (\eta t +u t^{1/2},t)}.
\end{aligned}
\end{equation}

The choice of $P_+ , P_-$ is such that $X_{\mathcal{L}^{-}\to (\eta t +u t^{1/2},t)}$ and $X_{\mathcal{L}^{+}\to (\eta t +u t^{1/2},t)}$ have the same leading order as $L_{\mathcal{L}\to (\eta t +u t^{1/2},t) }$. Denote by $X_{\mathcal{L}^{\pm}\to (\eta t +u t^{1/2},t)}^{\mathrm{resc}}$ the random variables
$X_{\mathcal{L}^{\pm}\to (\eta t +u t^{1/2},t)}$ rescaled as in \eqref{resc}.
Note now that $L_{\mathcal{L}^{+}\to P_{+}}$ and $L_{\mathcal{L}^{-}\to P_{-}} $ follow a simple central limit theorem. By this and law of large numbers \mbox{$L_{\mathcal{L}^{+}\to P_{+}}- (t-\frac{\eta t}{(1/\rho_{-} -1)^{2}})/\rho_{-}$} and \mbox{$L_{\mathcal{L}^{-}\to P_{-}}- (\eta t - \frac{t}{(1/(1-\rho_{+}) -1)^{2}})/(1-\rho_{+})$}
 converge, when divided by $t^{1/2},$ to $N_+ , N_{-}.$ On the other hand $L_{P_{+}\to (\eta t +u t^{1/2},t)}$ and $L_{P_{-}\to (\eta t +u t^{1/2},t)}$ have only $t^{1/3}$ fluctuations. After centering by their leading order and dividing by $t^{1/2}$, $L_{P_{+}\to (\eta t +u t^{1/2},t)}$ and $L_{P_{-}\to (\eta t +u t^{1/2},t)}$ converge in probability to $0$. This implies that
\begin{equation}\label{red}
\lim_{t\to\infty} \Pb(X_{\mathcal{L}^{\pm}\to (\eta t +u t^{1/2},t)}^{\mathrm{resc}} \leq s) = \Pb ( N_\pm \leq s).
\end{equation}
But $L_{\mathcal{L}^{+}\to P_{+}}$ and $L_{\mathcal{L}^{-}\to P_{-}}$ are by definition independent. Therefore we obtain for independent $N_+ , N_-$
\begin{equation}\label{prel}
\begin{aligned}
 \lim_{t\to \infty}\Pb\left( X_{\mathcal{L}^{-}\to (\eta t +u t^{1/2},t)}^{\mathrm{resc}} - X_{\mathcal{L}^{+}\to (\eta t +u t^{1/2},t)}^{\mathrm{resc}} >0\right)=\Pb(N_- - N_+ >0 ).
 \end{aligned}
\end{equation}
Finally, since $X_{\mathcal{L}^{\pm}\to (\eta t +u t^{1/2},t) }\leq L_{\mathcal{L}^{\pm}\to (\eta t +u t^{1/2},t)}$, Equation \eqref{red} together with Lemma~\ref{lemma4.1} imply that $L_{\mathcal{L}^{+}\to (\eta t +u t^{1/2},t) }^{\mathrm{resc}}- X_{\mathcal{L}^{+}\to (\eta t +u t^{1/2},t)}^{\mathrm{resc}}$ and $L_{\mathcal{L}^{-}\to (\eta t +u t^{1/2},t) }^{\mathrm{resc}}- X_{\mathcal{L}^{-}\to (\eta t +u t^{1/2},t)}^{\mathrm{resc}}$ converge to $0$ in probability, hence \eqref{prel} yields Theorem~\ref{berthm}.
\end{proof}

\section{Multipoint Distributions}\label{sectMultiPts}
In this section we consider an LPP model for which, in the TASEP picture there is a shock, and we ask the question of the joint distribution of the LPP times around the shock line, namely
\begin{equation}\label{multipoint}
\lim_{t\to \infty}\Pb\bigg(\bigcap_{k=1}^{m}\{L_{\mathcal{L}\to (\eta_0 t +u_k t^{1/3},t)}\leq \mu t+s_k t^{1/3}\}\bigg).
\end{equation}
To define the model, take $\beta >0$ and
\begin{equation}\label{pointtopoint}
\mathcal{L}^{+}=(\lfloor - \beta t \rfloor ,0) \quad \textrm{and}\quad \mathcal{L}^{-}=(0,\lfloor -\beta t \rfloor),
\end{equation}
as well as $\eta_0 =1$ and $\omega_{i,j} \sim \exp(1)$, see Figure~\ref{fig3a}.
\begin{figure}
\begin{center}
\psfrag{0}[r]{$0$}
\psfrag{Lp}[r]{$(-\beta t,0)$}
\psfrag{Lm}[r]{$(0,-\beta t)$}
\psfrag{P1}[c]{$P_1$}
\psfrag{P2}[c]{$P_2$}
\includegraphics[height=5cm]{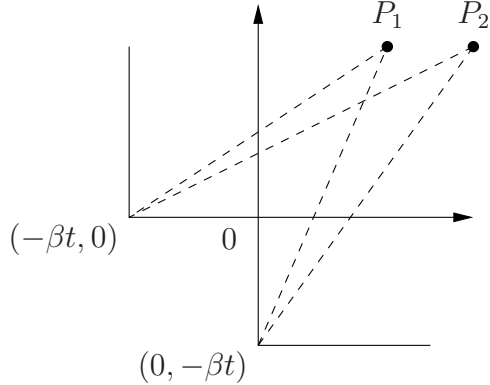}
\caption{The LPP model \eqref{pointtopoint}. The endpoints are at distance $\Or(t^{1/3})$: $P_k=(t+u_k t^{1/3},t)$.}
\label{fig3a}
\end{center}
\end{figure}
The physical picture in the background is the following: the LPP from ${\cal L}^\pm$ to different points at a distance $\Or(t)$ have fluctuations of order $\Or(t^{1/3})$ and the correlation length scales as $\Or(t^{2/3})$ due to KPZ scaling theory. Consider thus two different end-point at distance $u t^{2/3}$. Then, the law of large number from ${\cal L}^+$ and ${\cal L}^-$ change over that distance at first approximation linearly in $u t^{2/3}$ with two different prefactors (except along the diagonal). This means that in order to see the effect of both boundaries ${\cal L}^\pm$, we need to consider $u\sim t^{-1/3}$, otherwise one of the two LPP problem dominates completely the other. However, since on that small scale the rescaled LPP process should have Gaussian increments (as it is proven for the corresponding limit processes, the Airy processes~\cite{Ha07,CH11,QR12}) changes in the fluctuation of the two LPP problem will be $\Or(t^{1/6})$, thus irrelevant with respect to the $\Or(t^{1/3})$ fluctuations. This leads to the following result to be proven in Section~\ref{sect3.1}.
\begin{thm}\label{mltp} Let $u_ 1 < u_2 < \cdots < u_m$ be real numbers.
Consider the LPP model \eqref{pointtopoint}, $\mathcal{L}=\mathcal{L}^{+}\cup\mathcal{L}^{+}$ and points $P_k = (t+u_k t^{1/3}, t)$.
Then
\begin{equation}
\begin{aligned}
&\lim_{t \to \infty}\Pb\left( \bigcap_{k=1}^{m}\{L_{\mathcal{L}\to P_k}\leq (1+\sqrt{1+\beta})^{2}t+s_k t^{1/3}\}\right)
\\&= F_{\GUE}\left(\frac{\min_{k}\{s_k-\mu_+ u_k\}}{\sigma}\right)F_{\GUE} \left(\frac{\min_{k}\{s_k-\mu_- u_k\}}{\sigma}\right),
\end{aligned}
\end{equation}
with $\mu_+=1+1/\sqrt{1+\beta}$, $\mu_-=1+\sqrt{1+\beta}$, and $\sigma=\frac{(1+\beta+\sqrt{1+\beta})^{4/3}}{(1+\beta)^{1/6}}.$
\end{thm}

As explained in Section~\ref{gensect}, our Assumptions imply that $L_{\mathcal{L}^{+}\to (\eta_0 t +u_k t^{1/3},t)}$ and $L_{\mathcal{L}^{+}\to (\eta_0 t +u_k t^{1/3},t)}$ are asymptotically independent. Consequently the limit \eqref{multipoint} for $m=1$ is given by
\begin{equation}\label{1point}
\begin{aligned}
&\lim_{t\to\infty} \Pb\left(\frac{L_{\mathcal{L}\to (\eta_0 t +u_k t^{1/3},t)}-\mu t}{t^{1/3}}\leq s_k \right)\\
&=\lim_{t\to\infty} \Pb\left(\max\left\{\frac{L_{\mathcal{L}^{+}\to (\eta_0 t +u_k t^{1/3},t)}-\mu t}{t^{1/3}}, \frac{L_{\mathcal{L}^{-}\to (\eta_0 t +u_k t^{1/3},t)}-\mu t}{t^{1/3}} \right\} \leq s_k \right)\\&=G_1 (s_k) G_2 (s_k)
\end{aligned}
\end{equation}
as shown in Theorem 2.1 in~\cite{FN14}. For the particular model (\ref{pointtopoint}) the Assumptions 1,2,3 were checked in~\cite{FN14}, see Section 4.3.

To determine the limit \eqref{multipoint} for general $m$, the first step is to establish an extended no-crossing result, which guarantees asymptotic independence of the vectors $(L_{\mathcal{L}^{+}\to (\eta_0 t +u_k t^{1/3},t)})_{k=1,\ldots,m}$ and
$(L_{\mathcal{L}^{-}\to (\eta_0 t +u_k t^{1/3},t)})_{k=1,\ldots,m}$. To obtain then the limit law of each of these vectors, one then needs to control the rescaled local fluctuations as well,
\begin{equation}\label{locfluc}
\lim_{t\to \infty}\frac{L_{\mathcal{L}^{\pm}\to (\eta_0 t +u_k t^{1/3},t)}-L_{\mathcal{L}^{\pm}\to (\eta_0 t ,t)}}{t^{1/3}}.
\end{equation}

We give the extended no-crossing result in Proposition~\ref{extmult}, while \eqref{locfluc} has been obtained by Cator and Pimentel in~\cite{CP13} (they give a much more refined result actually), see Proposition~\ref{BM}.

For $P_k = (\eta(u_k) t, t)$ with $\eta(u_k) =1+u_k t^{-2/3}$, \eqref{eq5}, \eqref{eq6} of Assumption~2 hold for points $E_{k}^{+}$ on the line $\overline{(-\beta t,0) P_k}$, i.e., we may take
\begin{equation}\label{Ek}
E_k^{+}=(t-t^{\nu}(1+\beta)+u_k t^{1/3}-u_k t^{\nu-2/3}, t -t^{\nu}),
\end{equation}
see Figure~\ref{fig3} for an illustration.
\begin{figure}
\begin{center}
\psfrag{0}[r]{$0$}
\psfrag{Lp}[r]{$(-\beta t,0)$}
\psfrag{Lm}[r]{$(0,-\beta t)$}
\psfrag{P1}[c]{$P_1$}
\psfrag{P2}[c]{$P_2$}
\psfrag{E1}[c]{$E_1^+$}
\psfrag{E2}[c]{$E_2^+$}
\psfrag{D}[c]{$D_{\gamma_0}^2$}
\includegraphics[height=6cm]{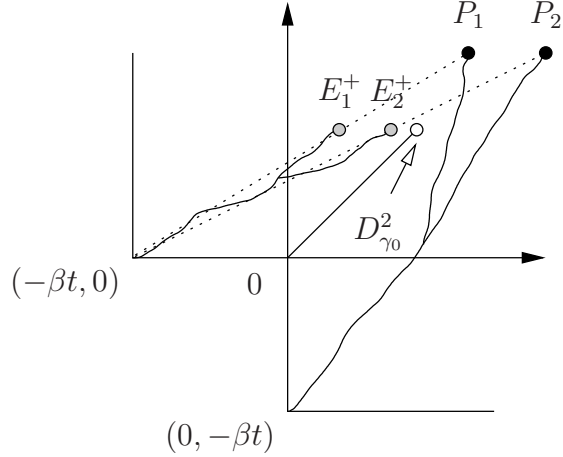}
\caption{The LPP model \eqref{pointtopoint} The maximizing paths $\pi^{\max}_{\mathcal{L}^{+}\to E_{1}^{+}}, \pi^{\max}_{\mathcal{L}^{+}\to E_{2}^{+}}, \pi^{\max}_{\mathcal{L}^{-}\to P_{1}},\pi^{\max}_{\mathcal{L}^{-}\to P_{2}}$ cross the line $\overline{(0,0)D_{\gamma_0}^{2}}$ with vanishing probability, where $D_{\gamma_0}^{2}=\gamma_0 P_2$ with $\gamma_0 = 1-t^{\nu-1}(1+\beta/2)$.}
\label{fig3}
\end{center}
\end{figure}

This, as well as the validity of Assumption~1, follow easily from the following theorem of Johansson for point-to-point last passage times,
which also identifies the limit distribution $G_1$ in this case.
\begin{prop}[Point-to-Point LPP: Convergence to $F_\GUE$, Theorem 1.6 in~\cite{Jo00b}] \label{propJohConvergence}
Let $0<\eta<\infty$. Then,
\begin{equation}
\begin{aligned}
\lim_{\ell\to\infty}\Pb\left(L_{(0,0)\to (\lfloor\eta\ell\rfloor,\lfloor\ell\rfloor)}\leq \mu_{\rm pp}\ell +s \sigma_\eta \ell^{1/3}\right)= F_\GUE(s)
\end{aligned}
\end{equation}
where $\mu_{\rm pp}=(1+\sqrt{\eta})^2$, and $\sigma_\eta=\eta^{-1/6}(1+\sqrt{\eta})^{4/3}$ with $F_\GUE$ the GUE Tracy Widom distribution function from random matrix theory.
\end{prop}

The following is the extended no-crossing result.
\begin{prop}\label{extmult}
Let $\nu \in (1/3,1), $ $u_1 <\cdots < u_m$ real numbers and take \mbox{$\gamma \in[0,1-t^{\nu-1}(1+\beta/2)]$}. Let $\mathcal{L}^{+}, \mathcal{L}^{-}$ be as in Theorem~\ref{mltp} and consider the points $D_{\gamma}^{m}=\gamma P_m= (\lfloor\gamma(t+u_m t^{1/3})\rfloor,\lfloor\gamma t\rfloor)$ and $E_k^{+}$ as in \eqref{Ek}. Then
\begin{equation}\label{ex1}
\lim_{t\to\infty}\Pb\bigg(\bigcup_{\gamma\in[0,1-t^{\nu-1}(1+\beta/2)]}\,\,\bigcup_{k=1}^{m} \{D_{\gamma}^{m}\in \pi^{\max}_{\mathcal{L}^{+}\to E_k^{+}}\}\bigg)=0
\end{equation}
and
\begin{equation}\label{ex2}
\lim_{t\to\infty}\Pb\bigg( \bigcup_{\gamma\in[0,1-t^{\nu-1}(1+\beta/2)]}\,\,\bigcup_{k=1}^{m}\{D_{\gamma}^{m}\in \pi^{\max}_{\mathcal{L}^{-}\to P_k}\}\bigg)=0.
\end{equation}
\end{prop}
The control over the local fluctuations for point-to-point problems is given in the following Proposition.
\begin{prop}[Corollary of Theorem 5 in~\cite{CP13}]\label{local} Let $s>0$ and $u\in \R$. We have in the sense of convergence in probability,
\begin{equation}\label{BM}
\lim_{t\to\infty}\frac{L_{(0,0)\to (st+ut^{1/3},t)}-\mu_s u t^{1/3} - L_{(0,0)\to (st,t)}}{t^{1/3}}=0,
\end{equation}
where $\mu_s =1+s^{-1/2}.$
\end{prop}
\begin{proof} This follows directly from the more refined Theorem 5 in~\cite{CP13}, which in particular tells that the process
\begin{equation}
\Delta_t (u)=\frac{L_{(0,0)\to (st+ut^{1/3},t)}-\mu_s u t^{1/3} - L_{(0,0)\to (st,t)}}{\mu_s t^{1/6}}
\end{equation}
converges to Brownian Motion in the sense of weak convergence of probability measures in the space of c\`{a}dl\`{a}g functions, implying \eqref{BM}.
\end{proof}
\begin{rem}Assumptions 1,2, the multipoint version of Assumption~3 (here Proposition~\ref{extmult}), and the control of the local fluctuations of $L_{\mathcal{L}^{\pm}\to (\eta t +ut^{1/3},t)}$ (here Proposition~\ref{local}) are sufficient to prove Theorem~\ref{mltp} for a general LPP model too.
\end{rem}

\subsection{Proof of Theorem~\ref{mltp}}\label{sect3.1}
To prove Theorem~\ref{mltp}, we need the multidimensional generalization of Lemma~\ref{lemma4.1}, and another elementary lemma from~\cite{CFP09}. As before, we denote by $``\Rightarrow"$ convergence in distribution and by $``\overset{d}{=}"$ equality in distribution. Let $X_n, \tilde{X}_n$ be random variables with values in $\R^m$. We say that $X_n \geq \tilde{X}_n$ if $X_n (k)\geq \tilde{X}_n (k)$ for each coordinate $X_n (k), \tilde{X}_n (k)$. We write $X_n \Rightarrow F$ where $F$ is a distribution function on $\R^{m}$, if for every continuity point $(s_1, \ldots, s_m)$ of $F$ we have that $\Pb(\cap_{k=1}^{m}\{X_n (k)\leq s_k \})$ converges to $F(s_1,\ldots,s_m)$. We say that $X_n -\tilde{X}_n$ converges to $0$ in probability as $n \to \infty$, if for every $\varepsilon >0$ we have $\Pb( || X_n-\tilde{X}_{n}||_\infty> \varepsilon)$ converges to $0$ as $n\to \infty$. Finally, for $X,Y$ in $\R^{m}$ we define $\max(X, Y):=Z$ coordinatewise, i.e., $Z(k)=\max\{X(k) , Y(k)\}$.
With this notation, we have the following lemmas.

\begin{lem}[Lemma 3.5 in~\cite{CFP09}, multidimensional version of Lemma~\ref{lemma4.1}]\label{lem35}
With the above notation for random variables in $\R^{m}$, Lemma~\ref{lemma4.1} holds true.
\end{lem}
\begin{lem}[Lemma 3.6 in~\cite{CFP09}]\label{lem36} Let $D_1, D_2, D_3 $ be probability distributions. Assume that $X_n \geq \tilde{X}_{n}$ and $X_n, \tilde{X}_n \Rightarrow D_1$, and equally that $Y_n \geq \tilde{Y}_{n}$ and $Y_n, \tilde{Y}_n \Rightarrow D_2$. Let $Z_n =\max\{X_n, Y_n\}$
and $\tilde{Z}_n =\max\{\tilde{X}_n, \tilde{Y}_n\}$. Then if $\tilde{Z}_n \Rightarrow D_3, $ also $Z_n \Rightarrow D_3$.
\end{lem}
To prove Theorem~\ref{mltp}, we reduce the problem to two independent random variables with the following two propositions.
\begin{prop}\label{propred4}
Let $\mu =(1+\sqrt{1+\beta})^{2},$ take $E_{k}^{+}$ from \eqref{Ek} and define $\R^m$-valued random variables
 \begin{equation}\begin{aligned}
X_n^1 &= \left(t^{-1/3}(L_{\mathcal{L}^{+}\to P_k}-\mu t)\right)_{k=1,\ldots,m},\\
\tilde{X}_n^1 &= \left(t^{-1/3}(L_{\mathcal{L}^{+}\to E_k^{+}}+L_{E^{+}_{k}\to P_k}-\mu t)\right)_{k=1,\ldots,m},\\
Y_n^{1}&=\tilde{Y}_{n}^{1}=\left(t^{-1/3}(L_{\mathcal{L}^{-}\to P_k}-\mu t)\right)_{k=1,\ldots,m}.
\end{aligned}\end{equation}
Then, for $D$ a probability distribution, if $\max(\tilde{X}_n^1 ,\tilde{Y}_{n}^{1})\Rightarrow D$, also $\max(X_n^1 ,Y_{n}^{1})\Rightarrow D$.
Furthermore, with $\sigma $ as in Theorem~\ref{mltp}
\begin{equation}
\begin{aligned}\label{red4}
X_{n}^{1}, \tilde{X}_n^1 &\Rightarrow F_{\GUE}\left(\frac{\min_k \{s_k -u_k(1+1/ \sqrt{1+\beta})\}}{\sigma}\right),\\
Y_{n}^{1}, \tilde{Y}_n^1 &\Rightarrow F_{\GUE}\left(\frac{\min_k \{s_k -u_k(1+ \sqrt{1+\beta})\}}{\sigma}\right).
\end{aligned}
\end{equation}
\end{prop}
\begin{proof}
We prove \eqref{red4} first and take the point $P=(t,t)$. For every \mbox{$k=1,\ldots,m$}, we write
\begin{equation}\label{kapp}
\frac{L_{\mathcal{L}^{+}\to P_k}-\mu t}{t^{1/3}}=\frac{L_{\mathcal{L}^{+}\to P}-\mu t}{t^{1/3}}+\frac{L_{\mathcal{L}^{+}\to P_k }-L_{\mathcal{L}^{+}\to P }}{t^{1/3}},
\end{equation}
and note $L_{\mathcal{L}^{+}\to P }\overset{d}{=}L_{(0,0)\to ((1+\beta)t,t) }$. Applying Proposition~\ref{local} with $s=1+\beta$ and Proposition~\ref{propJohConvergence} to \eqref{kapp} we obtain $X_{n}^{1}\Rightarrow F_{\GUE}\left(\frac{\min_k \{s_k -u_k(1+1/ \sqrt{1+\beta})\}}{\sigma}\right).$
Similarily, for $E^{+}=(t-t^{\nu}(1+\beta),t-t^{\nu})$, we decompose
\begin{equation}
\begin{aligned}\label{kennyp}
&\frac{L_{\mathcal{L}^{+}\to E_k^{+}}+L_{E^{+}_{k}\to P_k}-\mu t}{t^{1/3}}\\
=&\frac{L_{\mathcal{L}^{+}\to E^{+}}-\mu (t-t^{\nu})}{t^{1/3}}+\frac{L_{\mathcal{L}^{+}\to E_k^{+}}-L_{\mathcal{L}^{+}\to E^{+}}}{t^{1/3}}+\frac{L_{E^{+}_k\to P_k}-\mu t^{\nu}}{t^{1/3}}.
\end{aligned}
\end{equation}
Since the points $E_{k}^{+}$ satisfy Assumption~2 of our LPP model, the last summand in \eqref{kennyp} converges to $0$, while the second summand converges to $u_k (1+1/\sqrt{1+\beta})$ by Proposition~\ref{local}. Furthermore, we have that $\Pb\left(\frac{L_{\mathcal{L}^{+}\to E^{+}}-\mu (t-t^{\nu})}{t^{1/3}}\leq s\right)$
converges to $F_\GUE (s)$ by Proposition~\ref{propJohConvergence}. So in total we get $\tilde{X}_{n}^{1}\Rightarrow F_{\GUE}\left(\frac{\min_k \{s_k -u_k(1+1/ \sqrt{1+\beta})\}}{\sigma}\right)$.
Finally, replacing $\mathcal{L}^{+}$ by $\mathcal{L}^{-}$ in \eqref{kapp} and noting that $L_{\mathcal{L}^{-}\to P }\overset{d}{=}L_{(0,0)\to ((1+\beta)^{-1}\ell,\ell) } $ with $\ell=(1+\beta)t$, by the same argument we get $Y_{n}^{1} \Rightarrow F_{\GUE}\left(\frac{\min_k \{s_k -u_k(1+ \sqrt{1+\beta})\}}{\sigma}\right)$. The first assertion of the Proposition follows now from Lemma~\ref{lem36}, since $X_{n}^{1}\geq \tilde{X}_n^1$ and $Y_n^{1}=\tilde{Y}_{n}^{1}$.
\end{proof}

In the following, we denote for sets $A,B$ by $\tilde{L}_{A\to B}$ the LPP time from $A$ to $B$ with the difference that the maximum in \eqref{LPP} is only taken over up-right paths that do not contain any point $D_{\gamma}^{m}=(\lfloor\gamma(t+u_m t^{1/3})\rfloor,\lfloor\gamma t\rfloor)$, $\gamma \leq 1-t^{\nu-1}(1+\beta/2)$.

\begin{prop}\label{propred5}Let $D$ be a probability distribution and
define $X_{n}^{2}=\tilde{X}_{n}^{1},$ $Y_{n}^{2}=Y_{n}^{1}$ and
\begin{equation}
\begin{aligned}
\tilde{X}_n^{2} &= \left(t^{-1/3}(\tilde{L}_{\mathcal{L}^{+}\to E_k^{+}}+L_{E_{k}^{+}\to P_k}-\mu t)\right)_{k=1,\ldots,m},\\
\tilde{Y}_{n}^{2} &=\left(t^{-1/3}(\tilde{L}_{\mathcal{L}^{-}\to P_k}-\mu t)\right)_{k=1,\ldots,m}.
\end{aligned}
\end{equation}
Then if $\max(\tilde{X}_n^2 ,\tilde{Y}_{n}^{2})\Rightarrow D$, also $\max(X_n^2 ,Y_{n}^{2})\Rightarrow D$.
Furthermore, \begin{equation}
\begin{aligned}\label{red5}
\tilde{X}_n^2 &\Rightarrow F_{\GUE}\left(\frac{\min_k \{s_k -u_k(1+1/ \sqrt{1+\beta})\}}{\sigma}\right),\\
\tilde{Y}_n^2 &\Rightarrow F_{\GUE}\left(\frac{\min_k \{s_k -u_k(1+ \sqrt{1+\beta})\}}{\sigma}\right).
\end{aligned}
\end{equation}
\end{prop}
\begin{proof}
We prove \eqref{red5} first.
First we note that
$X_{n}^{2}\geq \tilde{X}_n^{2} $
and $Y_{n}^{2}\geq \tilde{Y}_n^2$.
Now, for $\varepsilon >0$,
\begin{equation}
\Pb(||X_{n}^{2}-\tilde{X}_{n}^{2}||_{\infty}>\varepsilon )\leq \Pb\bigg(\bigcup_{\gamma\in[0,1-t^{\nu-1}(1+\beta/2)]}\,\,\bigcup_{k=1}^{m} \{D_{\gamma}^{m}\in \pi^{\max}_{\mathcal{L}^{+}\to E_k^{+}}\}\bigg),
\end{equation}
which goes to $0$ as $t\to\infty$ by Proposition~\ref{extmult}. Thus, by \eqref{red4} and Lemma~\ref{lem35}, it follows that $\tilde{X}_n^2 \Rightarrow F_{\GUE}\left(\frac{\min_k \{s_k -u_k(1+1/ \sqrt{1+\beta})\}}{\sigma}\right)$. An analogous proof shows that $\tilde{Y}_n^2 \Rightarrow F_{\GUE}\left(\frac{\min_k \{s_k -u_k(1+ \sqrt{1+\beta})\}}{\sigma}\right)$. The first assertion of the proposition follows now from \eqref{red4},\eqref{red5} and Lemma~\ref{lem36}.
\end{proof}
With the reductions from Propositions~\ref{propred4},~\ref{propred5} at hand, we can now proceed to the proof of Theorem~\ref{mltp}.
\begin{proof}[Proof of Theorem~\ref{mltp}]
For the sake of brevity, we define
\begin{equation}
\begin{aligned}
F(s_1,\ldots,s_m)= &F_{\GUE}\left(\frac{\min_k \{s_k -u_k(1+1/ \sqrt{1+\beta})\}}{\sigma}\right) \\
\times &F_{\GUE}\left(\frac{\min_k \{s_k -u_k(1+ \sqrt{1+\beta})\}}{\sigma}\right)
\end{aligned}
\end{equation}
and the rescaled LPP times
\begin{equation}
\begin{aligned}
\tilde{L}_{\mathcal{L}^{+}\to E^{+}_{k}}^{\mathrm{resc}}=\frac{\tilde{L}_{\mathcal{L}^{+}\to E^{+}_{k}}-\mu (t-t^{\nu})}{t^{1/3}},
\quad \tilde{L}_{\mathcal{L}^{-}\to P_{k}}^{\mathrm{resc}}=\frac{\tilde{L}_{\mathcal{L}^{-}\to P_{k}}-\mu t}{t^{1/3}}.
\end{aligned}
\end{equation}
First of all, note that
\begin{equation}
\left(t^{-1/3}(L_{\mathcal{L}\to P_k }-\mu t)\right)_{k=1,\ldots,m}=\max \{X_{n}^{1}, Y_{n}^{1}\},
\end{equation}
so that by Propositions~\ref{propred4},~\ref{propred5} we have to show
\begin{equation}
\begin{aligned}\label{redproof}
\max\{\tilde{X}_{n}^{2}, \tilde{Y}_{n}^{2}\}&\Rightarrow F(s_1,\ldots,s_m)\quad\textrm{as }t\to\infty.
\end{aligned}
\end{equation}
Define the random variable $X_{t}^{k}$ via
\begin{equation}
\frac{L_{E_{k}^{+}\to P_{k}} -\mu t^{\nu}}{t^{1/3}}=\frac{1}{t^{(1-\nu)/3}} X_{t}^{k}.
\end{equation}
Then $\frac{1}{t^{(1-\nu)/3}} X_{t}^{k}$ vanishes as $t\to \infty$. This fact together with Lemma~\ref{lem35} and \eqref{red5} implies
\begin{equation}\label{red6}
\left( \tilde{L}_{\mathcal{L}^{+}\to E^{+}_{k}}^{\mathrm{resc}}\right)_{k=1,\ldots,m}\Rightarrow F_{\GUE}\left(\frac{\min_k \{s_k -u_k(1+1/ \sqrt{1+\beta})\}}{\sigma}\right)
\end{equation}
as $t\to\infty$.

Let now $\varepsilon >0$ and take $R>0$ such that $\Pb(\cap_{k=1}^{m}\{|X_{t}^{k}|\leq R\})\geq 1- \varepsilon$ and define $A_{R}=\cap_{k=1}^{m}\{|X_{t}^{k}|\leq R\}$.
In particular it holds
\begin{equation}\label{easy}
\begin{aligned}
|\Pb(\cap_{k=1}^{m}&\{\max{\{\tilde{X}_{n}^{2}(k), \tilde{Y}_{n}^{2}(k)\}\leq s_k}\})\\&-\Pb(\cap_{k=1}^{m}\{\max{\{\tilde{X}_{n}^{2}(k), \tilde{Y}_{n}^{2}(k)\}\leq s_k}\}\cap A_R)|\leq \varepsilon.
\end{aligned}
\end{equation}
We then have
\begin{align}\label{1ter}
&\Pb\bigg(\bigcap_{k=1}^{m}\left(
\{ \tilde{L}_{\mathcal{L}^{+}\to E^{+}_{k}}^{\mathrm{resc}}+t^{(\nu-1)/3} R\leq s_{k}\}\cap \{ \tilde{L}_{\mathcal{L}^{-}\to P_{k}}^{\mathrm{resc}}\leq s_{k}\}\right)\bigg)-\varepsilon
\\&\label{2ter}\leq \Pb\bigg(\bigcap_{k=1}^{m}\left(
\{ \tilde{L}_{\mathcal{L}^{+}\to E^{+}_{k}}^{\mathrm{resc}}+t^{(\nu-1)/3} X_{t}^{k}\leq s_{k}\}\cap A_R \cap \{ \tilde{L}_{\mathcal{L}^{-}\to P_{k}}^{\mathrm{resc}}\leq s_{k}\}\right)\bigg)
\\&\leq
\Pb\bigg(\bigcap_{k=1}^{m}\left(
\{ \tilde{L}_{\mathcal{L}^{+}\to E^{+}_{k}}^{\mathrm{resc}}-t^{(\nu-1)/3} R\leq s_{k}\}\cap A_R \cap \{ \tilde{L}_{\mathcal{L}^{-}\to P_{k}}^{\mathrm{resc}}\leq s_{k}\}\right)\bigg)
\\&\label{6ter}
\leq \Pb\bigg(\bigcap_{k=1}^{m}\left(
\{ \tilde{L}_{\mathcal{L}^{+}\to E^{+}_{k}}^{\mathrm{resc}}-t^{(\nu-1)/3} R\leq s_{k}\}\cap \{ \tilde{L}_{\mathcal{L}^{-}\to P_{k}}^{\mathrm{resc}}\leq s_{k}\}\right)\bigg)
\end{align}
Now $\left(\tilde{L}_{\mathcal{L}^{+}\to E^{+}_{k}}^{\mathrm{resc}}\right)_{k=1,\ldots ,m}$ and
$\left(\tilde{L}_{\mathcal{L}^{-}\to P_{k}}^{\mathrm{resc}}\right)_{k=1,\ldots ,m}$ are independent random variables, since the $x$ coordinate of all $E_{k}^{+}$ is smaller than the $x$ coordinate of $D_{\gamma}^{m}$ for $\gamma=1-t^{\nu-1}(1+\beta/2)$ and $t$ large enough.
Hence, by \eqref{red5} and \eqref{red6} there is a $t_{0}$ such that for $t>t_0$
\begin{equation}
F(s_{1},\ldots,s_{m})-2\varepsilon \leq \eqref{1ter}\leq \eqref{2ter} \leq \eqref{6ter}\leq F(s_1, \ldots,s_m)+\varepsilon.
\end{equation}
Hence applying \eqref{easy} to \eqref{2ter} yields
\begin{equation}
|\Pb(\cap_{k=1}^{m}\{\max{\{\tilde{X}_{n}^{2}(k), \tilde{Y}_{n}^{2}(k)\}\leq s_k}\})-F(s_1,\ldots,s_m)|\leq 3 \varepsilon
\end{equation}
for $t$ large enough, thus proving \eqref{redproof}.
\end{proof}

\subsection{Proof of Proposition~\ref{extmult}}
The no-crossing result \eqref{ex1} will be deduced from the validity of Assumption~3 established in~\cite{FN14} by a soft argument, see also Figure~\ref{fig3}. To prove \eqref{ex2}, we use an extension of the strategy used in~\cite{FN14}, which is based on moderate deviation estimates for LPP times, given in the following two propositions.
\begin{prop}[Proposition 4.2 in~\cite{FN14}]\label{devone}
Let $0<\eta<\infty$. Then for given $\ell_0>0$ and $s_0 \in \mathbb{R}$, there exist constants $C,c>0$ only dependent on $\ell_0,s_0$ such that for all $\ell\geq \ell_0$ and $s\geq s_0$ we have
\begin{equation} \begin{aligned} \label{moddef}
\Pb\left(L_{(0,0)\to (\lfloor\eta\ell\rfloor,\lfloor\ell\rfloor)}> \mu_{\rm pp}\ell +\ell^{1/3}s\right)\leq C \exp(-c s),
\end{aligned}\end{equation}
where $\mu_{\rm pp}=(1+\sqrt{\eta})^{2}$.
\end{prop}
\begin{prop}[Proposition~4.3 in~\cite{FN14}]\label{devtwo} Let \mbox{$0<\eta<\infty$} and \mbox{$\mu_{\rm pp}=(1+\sqrt{\eta})^2$}. There exist positive constants $s_0,\ell_0,C,c$
such that for $s\leq -s_0,$ $\ell\geq \ell_0$,
\begin{equation} \begin{aligned}
\Pb\left(L_{(0,0)\to (\lfloor\eta\ell\rfloor,\lfloor\ell\rfloor)}\leq \mu_{\rm pp} \ell +s\ell^{1/3}\right)\leq C\exp(-c|s|^{3/2}).
\end{aligned}\end{equation}
\end{prop}
The following proposition shows that if $\pi^{\max}_{\mathcal{L}^{-}\to P_k}$ passes through a point $\gamma P_l$, then $L_{\mathcal{L}^{-}\to P_k}$ is unlikely to be larger than the leading order of \mbox{$L_{\mathcal{L}^{-}\to \gamma P_l}+ L_{\gamma P_l \to P_k}$} plus a $\mathcal{O}(t^{1/3})$ term.
\begin{prop}\label{largedev}
Let $r_1, r_2 \in \R,$ $\nu \in (1/3,1)$ and $\varepsilon=ct^{\nu-1},c>0$. Define points $D_{\gamma}(r_1)=(\gamma t+\gamma r_1 t^{1/3},\gamma t)$ for $\gamma \in [0, 1-t^{\nu-1}(1+\beta/2)]$ and \mbox{$P(r_2)=(t+r_2 t^{1/3},t)$}.
Define
\begin{equation}
E_{r_1, r_2,\gamma}=\{L_{\mathcal{L}^{-}\to D_{\gamma}(r_1)}\leq (\mu_\gamma
+\varepsilon/2)t \}\cap \{L_{ D_{\gamma}(r_1)\to P(r_2)}\leq (\mu_{\mathrm{pp},\gamma}+\varepsilon/2)t\}.
\end{equation}
where $\mu_{\mathrm{pp},\gamma}=4(1-\gamma),$ $\mu_{\gamma}=2\gamma+\beta+2\sqrt{\gamma(\gamma+\beta)}.$
Then for some constants $\tilde{c}, \tilde{C}>0$\begin{equation}\label{strike1}
\Pb\bigg(\bigcup_{0\leq\gamma \leq 1-t^{\nu-1}(1+\beta/2)} \Omega\setminus E_{r_1, r_2,\gamma} \bigg)\leq \tilde{C}e^{-\tilde{c}t^{\nu-1/3}}.
\end{equation}
\end{prop}
\begin{proof}
We denote by $\mu_1$ the constant $\mu_{\mathrm{pp}}$ from Proposition~\ref{propJohConvergence} for
\begin{equation}
L_{\mathcal{L}^{-}\to D_{\gamma}(r_1)}\overset{d}{=} L_{(0,0)\to (\ell(1+\beta/\gamma+\gamma^{-2/3}r_1 \ell^{-2/3}),\ell)}=:L_{(0,0)\to (\eta_1 \ell,\ell)}
\end{equation}
with $\ell=\gamma t$ and by $\mu_2$ the constant $\mu_{\mathrm{pp}}$ from Proposition~\ref{propJohConvergence} for
\begin{equation}
L_{D_{\gamma}(r_1)\to P(r_2)}\overset{d}{=}
L_{(0,0)\to (\tilde{\ell}+(r_2-r_1 \gamma)(\tilde{\ell}/(1-\gamma))^{1/3} , \tilde{\ell})}=:L_{(0,0)\to (\eta_2 \tilde{\ell},\tilde{\ell})}
\end{equation}
with $\tilde{\ell}=(1-\gamma)t$. A simple computations gives
\begin{equation}
\begin{aligned}
\mu_1 \ell&=\gamma t \left(1+\sqrt{1+\beta/\gamma+r_2 t^{-2/3}}\right)^{2}= (\mu_{\gamma}+\mathcal{O}(t^{-2/3}))t,\\
\mu_2 \tilde{\ell} &=(1-\gamma)t\left(1+\sqrt{1+(r_2-r_1\gamma)t^{-2/3}/(1-\gamma)}\right)^{2}=(\mu_{\mathrm{pp},\gamma}+\mathcal{O}(t^{-2/3}))t.
\end{aligned}
\end{equation}
Since $\nu >1/3$ we thus have for some $d, C_1, c_1 >0$ and $t$ large enough
\begin{equation}
\begin{aligned}
\Pb((E_{r_1, r_2, \gamma})^{c} )&\leq \Pb(L_{(0,0)\to (\eta_1 \ell,\ell)}> \mu_1 \ell+d (\ell/\gamma)^{\nu})\\&+\Pb(L_{(0,0)\to (\eta_2 \tilde{\ell},\tilde{\ell})}> \mu_2 \tilde{\ell}+d (\tilde{\ell}/(1-\gamma))^{\nu})
\\&\leq C_1e^{-c_1t^{\nu-1/3}},
\end{aligned}
\end{equation}
where for the last inequality we used Proposition~\ref{devone}.
Since there are only $\mathcal{O}(t)$ many events $(E_{r_1, r_2, \gamma})^{c}$, \eqref{strike1} follows.
\end{proof}
The next proposition shows that the leading order of $L_{\mathcal{L}^{-}\to \gamma P_l}+ L_{\gamma P_l \to P_k}$ is $\mathcal{O}(t^{\nu})$ smaller than the one of $L_{\mathcal{L}^{-}\to P_l}$.
\begin{prop}\label{lead}
Let $\beta >0$, $\nu \in (1/3,1)$ and $\gamma \in [0, 1-t^{\nu-1}(1+\beta/2)]$. Then, for $t$ large enough, we have with $\varepsilon =Ct^{\nu-1}$
\begin{equation}
\frac{(\mu_{\mathrm{pp},\gamma}+\mu_{\gamma}+\varepsilon-\mu)t}{t^{1/3}}\leq -Ct^{\nu-1/3},
\end{equation}
where $\mu_{\mathrm{pp},\gamma}=4(1-\gamma),$ $\mu_{\gamma}=2\gamma+\beta+2\sqrt{\gamma(\gamma+\beta)},$ $\mu=(1+\sqrt{1+\beta})^{2}$ and $C=C(\beta)=-\frac{1}{2}(2-2\sqrt{1+\beta}+\frac{\beta}{\sqrt{1+\beta}})>0$.
\end{prop}
\begin{proof}
For $Z(\gamma)=\mu_{\gamma}+\mu_{\mathrm{pp},\gamma}$ we have $Z(0)<Z(1)$ and $Z^{\prime}(\gamma)\neq 0$, thus $Z$ is monotonely increasing in $[0,1]$. Let $\gamma_{0}=1-t^{\nu-1}(1+\beta/2)$.
Then \begin{equation}
\begin{aligned}
Z(\gamma_{0})&=4+\beta+2\gamma_{0}(-1+\sqrt{1+\beta/\gamma_{0}})
\\&=4+\beta +2(1-t^{\nu-1}(1+\beta/2))\left(-1+\sqrt{1+\beta}+\frac{t^{\nu-1}\beta (1+\beta/2)}{2\sqrt{1+\beta}}+\mathcal{O}(t^{2\nu-2})\right)
\\&=2+\beta+2\sqrt{1+\beta}+t^{\nu-1}(1+\beta/2)\left(-2\sqrt{1+\beta}+2+\frac{\beta}{\sqrt{1+\beta}}\right)+\mathcal{O}(t^{2\nu-2})
\\&=\mu-2(1+\beta/2)Ct^{\nu-1}+\mathcal{O}(t^{2\nu-2}).
\end{aligned}
\end{equation}
In particular, $(Z(\gamma_0)+\varepsilon-\mu)t^{2/3}\leq -Ct^{\nu-1/3}$ for $t$ large enough.
\end{proof}
We can now proceed to the proof of Proposition~\ref{extmult}.
\begin{proof}[Proof of Proposition~\ref{extmult}]
Define the events
\begin{equation}
F_{k}^{m,+}=\bigcap_{\gamma \in [0, 1-t^{\nu-1}(1+\beta/2)]}\{D_{\gamma}^{m}\notin \pi_{\mathcal{L}^{+}\to E_{k}^{+}}^{\max} \}.
\end{equation}
The fact that $\Pb(F_{m}^{m,+})$ converges to $1$ as $t\to\infty$ is precisely the content of Assumption~3 in our LPP model \eqref{pointtopoint} with $\eta=1+u_m t^{-2/3}$, which was checked in~\cite{FN14}, Section~4.3, Proof of Corollary~2.4. Now we have \mbox{$F_{m}^{m,+}\subseteq F_{k}^{m,+}$} for $k=1,\ldots,m-1$, since if $\pi_{\mathcal{L}^{+}\to E_{k}^{+}}^{\max} $ has to branch from $\pi_{\mathcal{L}^{+}\to E_{m}^{+}}^{\max} $ to contain a point $D_{\gamma}^{m}$, $\pi_{\mathcal{L}^{+}\to E_{k}^{+}}^{\max} $ has to cross $\pi_{\mathcal{L}^{+}\to E_{m}^{+}}^{\max} $ again to reach $E_{k}^{+}$, which is impossible. Consequently,
\begin{equation}
\lim_{t\to\infty}\Pb\bigg(\bigcup_{k=1}^{m}(F_{k}^{m,+})^{c}\bigg)=0,
\end{equation}
which is exactly \eqref{ex1}.
Next define the events
\begin{equation}
I_{k,m, \gamma}=\{D_{\gamma}^{m}\in \pi_{\mathcal{L}^{-}\to P_{k}}^{\max} \}\quad\textrm{and}\quad F_{k}^{m,-}=\bigcap_{\gamma \in [0, 1-t^{\nu-1}(1+\beta/2)]}(I_{k,m, \gamma})^{c}.
\end{equation}
Now we consider the events $E_{r_1, r_2,\gamma}$ from Proposition~\ref{largedev} with $r_1=u_m$ and $r_2=u_k$ and take $\varepsilon=Ct^{\nu-1}$ with $C=C(\beta)$ from Proposition~\ref{lead}. We then bound
\begin{equation}
\Pb(I_{k,m, \gamma})\leq \Pb\bigg(I_{k,m, \gamma}\cap \bigcap_{\gamma} E_{u_m, u_k,\gamma} \bigg)+\Pb\bigg(\bigcup_{\gamma} (E_{u_m, u_k,\gamma})^{c}\bigg).
\end{equation}
 By Proposition~\ref{largedev}, we have $\Pb\big(\bigcup_{\gamma} (E_{u_m, u_k,\gamma})^{c}\big)\leq \tilde{C}e^{-\tilde{c}t^{\nu-1/3}}$.
Now for \mbox{$\omega \in I_{k,m, \gamma}\cap \bigcap_{\gamma} E_{u_m, u_k,\gamma}$} we have
\begin{equation}
L_{\mathcal{L}^{-}\to P_k}(\omega)\leq (\mu_{\mathrm{pp},\gamma}+\mu_\gamma+\varepsilon)t.
\end{equation}
Consequently, by Proposition~\ref{lead} we have
\begin{equation}
I_{k,m, \gamma}\cap \bigcap_{\gamma} E_{r_1, r_2,\gamma} \subseteq
\{L_{\mathcal{L}^{-}\to P_k}\leq \mu t-C t^{\nu-1/3}\}.
\end{equation}
Using Proposition~\ref{devtwo}, we thus see for $t$ large enough and some constants $\tilde{C}_1, \tilde{c}_1 >0$
\begin{equation}
\Pb\bigg(I_{k,m, \gamma}\cap \bigcap_{\gamma} E_{r_1, r_2,\gamma}\bigg)\leq \tilde{C}_1 e^{-\tilde{c}_1 t^{\nu-1/3}}.
\end{equation}
Putting all together, we obtain
\begin{equation}
\Pb(I_{k,m, \gamma})\leq \tilde{C}_2 e^{-\tilde{c}_2 t^{\nu-1/3}}
\end{equation}
for some constants $\tilde{C}_2, \tilde{c}_2 >0$. Since there are only $\mathcal{O}(t)$ many events $I_{k,m, \gamma}$, this implies that
\begin{equation}
\Pb\left(F_{k}^{m,-}\right)\geq 1-\mathcal{O}(t) \tilde{C}_2 e^{-\tilde{c}_2 t^{\nu-1/3}}
\end{equation}
which converges to $1$ as $t\to\infty$, for all $k=1,\ldots,m$, proving \eqref{ex2}.
\end{proof}


\end{document}